\documentclass[11pt]{article}

\usepackage[utf8]{inputenc}

\usepackage{amsmath,amsfonts,amsthm,amssymb,color}
\usepackage[usenames,dvipsnames,svgnames,table]{xcolor}
\definecolor{darkgreen}{rgb}{0.0,0,0.9}
\usepackage{mathtools}
\usepackage{authblk}
\usepackage{fullpage}
\usepackage{parskip}
\usepackage{comment}
\usepackage{tikz}
\usepackage{bbm}
\usepackage{dsfont}
\usepackage[sc]{mathpazo}
\usepackage[basic]{complexity}
\usepackage{algorithm2e}
\usepackage[colorlinks=true,
citecolor=OliveGreen,linkcolor=BrickRed,urlcolor=BrickRed,pdfstartview=FitH]{hyperref}
\usepackage[capitalize,nameinlink]{cleveref}
\usepackage{tcolorbox}
\usepackage[short]{optidef}

\newtcolorbox{wbox}
{
	colback  = white,
}

\SetKwInOut{Input}{Input}
\SetKwInOut{Output}{Output}
\SetKwFunction{Uncross}{\textsc{Uncross}}
\SetKwFunction{MergeUncross}{\textsc{MergeUncross}}
\SetKwFunction{PerfectMatching}{\textsc{PerfectMatching}}
\SetKwBlock{InParallel}{in parallel do}{end}
\SetKwFor{ParallelFor}{for}{in parallel do}{end}

\let\R\relax
\newcommand*{\R}{\mathbb{R}}

\newcommand*{\Rplus}{\mathbb{R_+}}

\newcommand*{\suppress}[1]{}

\newcommand*{\cR}{\mathcal{R}}
\newcommand*{\ccP}{\mathcal{P}}

\makeatletter
\def\thm@space@setup{%
	\thm@preskip= 10pt
	\thm@postskip=\thm@preskip 
}
\makeatother

\makeatletter
\renewcommand{\paragraph}{%
	\@startsection{paragraph}{4}%
	{\z@}{5pt}{-1em}%
	{\normalfont\normalsize\bfseries}%
}
\makeatother

\newtheorem{theorem}{Theorem}
\newtheorem{lemma}{Lemma}

\newtheorem{corollary}{Corollary}

\newtheorem{fact}{Fact}

\theoremstyle{definition}
\newtheorem{definition}{Definition}

\newtheorem{remark}{Remark}

\newtheorem{example}{Example}

\newenvironment{fminipage}%
{\begin{Sbox}\begin{minipage}}%
		{\end{minipage}\end{Sbox}\fbox{\TheSbox}}

\newcommand\QQ{\boldsymbol{\mathit{Q}}}

\newcommand\ZZ{\boldsymbol{\mathit{Z}}}

\newcommand*{\ov}[1]{\bar{#1}}
\newcommand*{\un}[1]{\underline{#1}}
\newcommand{\profit}{\mbox{\rm profit}}
\newcommand{\worth}{\mbox{\rm worth}}
\newcommand{\surplus}{\mbox{\rm surplus}}
\newcommand{\spread}{\mbox{\rm spread}}

\title{LP-Duality Theory and the Cores of Games}

\author[1]{Vijay V.~Vazirani\footnote{Supported in part by NSF grant CCF-2230414.}}

\affil[1]{University of California, Irvine}

\date{}

\begin{document}
	\maketitle
	
	\begin{abstract}
LP-duality theory has played a central role in the study of the core, right from its early days to the present time. However, despite the extensive nature of this work, basic gaps still remain. We address these gaps using the following building blocks from LP-duality theory: 
\begin{enumerate}
	\item Total unimodularity (TUM). 
	\item Complementary slackness conditions and strict complementarity.
\end{enumerate}

Our exploration of TUM leads to defining new games, characterizing their cores and giving novel ways of using core imputations to enforce constraints that arise naturally in applications of these games. The latter include:
\begin{enumerate}
\item Efficient algorithms for finding {\em min-max fair, max-min fair and equitable  core imputations}.   

\item {\em Encouraging diversity and avoiding over-representation} in a generalization of the assignment game.
\end{enumerate}

Complementarity enables us to prove new properties of core imputations of the assignment game and its generalizations. 
	\end{abstract}

\bigskip
\bigskip
\bigskip
\bigskip
\bigskip
\bigskip
\bigskip
\bigskip
\bigskip
\bigskip
\bigskip
\bigskip
\bigskip
\bigskip
\bigskip
\bigskip
\bigskip
\bigskip
\bigskip
\bigskip
\bigskip
\bigskip

\pagebreak

\section{Introduction}
\label{sec.intro}

The core\footnote{The core contains all possible ways of distributing the total profit of a game among individual agents in such a way that the grand coalition remains intact, i.e., a sub-coalition will not be able to generate more profits by itself and therefore has no incentive to secede from the grand coalition.} is a quintessential solution concept, for ``fair'' profit sharing, in cooperative game theory, and LP-duality theory has played a central role in its study, right from the early works of Bondareva \cite{Bondareva1963some} and Shapley \cite{Shapley1965balanced} to the present time. However, our contention is that basic gaps still remain. In Section \ref{sec.gaps}, we point out these gaps and in Section \ref{sec.plan} we outline our approach for filling them; in particular, we define new games, having important applications, and characterize their cores. We also show how to enforce constraints, that arise naturally in these applications, via core imputations, see Section \ref{sec.enforcing}.

\subsection{Enforcing Special Constraints via Core Imputations}
\label{sec.enforcing}

{\bf The notions of min-max fair, max-min fair and equitable  core imputations:}
The core is considered the ``gold standard'' for fair profit sharing. The criterion of fairness satisfied by a core imputation is the following: for each of exponentially many sub-coalitions, it  gives at least as much profit as the sub-coalition can make by itself. An interesting side effect  is the following: under core imputations, the profit allocated to an agent is consistent with the value he/she brings to the various sub-coalitions he/she belongs to, i.e., it is consistent with his/her negotiating power, e.g., see Section 3.2.3 in \cite{Transfers}. However, there is much variability across core imputations and specific applications may call for core imputations with specific properties. 

Perhaps the most prominent example of this phenomenon arises in the stable matching game of Gale and Shapley \cite{GaleS}. The core of this game is the set of all stable matchings --- each ensures that no coalition formed by one agent from each side of the bipartition has an incentive to secede. Two extreme matchings in this set are the top and bottom elements in the lattice of stable matchings; these can be found in polynomial time via the deferred acceptance algorithm \cite{GaleS}. Each of these matchings  maximally favors one side and disfavors the other side of the bipartition. However, several applications of this game call for stable matchings which are equitable to both sides, and the question of finding such matchings efficiently has garnered much attention, e.g., see \cite{Irving-egalitarian, Teo1998geometry, Two-sided}. Another prominent example of this phenomenon is core-selecting auctions \cite{Milgrom-Day2007core}. 
 
 For the assignment game, Shapley and Shubik \cite{Shapley1971assignment} showed that if the core contains more than one imputation, then it contains two extreme imputations, each maximally favoring one side and disfavoring the other side of the bipartition, much the same way \footnote{Note that the two problems differ in a fundamental way: in stable matching, agents have ordinal utilities and in the assignment game, utilities are cardinal, given by weights on the edges of the given bipartite graph.} as the two extreme stable  matchings, see Theorem \ref{thm.extreme}. Can we find core imputations which are ``more fair'' to all agents? We answer this question by defining the notions of {\em min-max fair, max-min fair and equitable  core imputations} for the assignment game and giving efficient algorithms for finding them, see Section \ref{sec.fair}. Furthermore, we do the same for all generalizations of the assignment game studied in this paper. 
 
Another game which calls for similar notions is the max-flow game. The combinatorial problem underlying this game is the maximum flow problem, which is well known to have widespread applications\footnote{Indeed, it is among the three most prominent problems in combinatorial optimization \cite{LP.book, Sch-book}, along with maximum matching and linear programming.}.  The max-flow game was defined by Kalai and Zemel \cite{Kalai1982totally}, and they also gave one core imputation for it -- corresponding to a minimum cut. Observe that this imputation allocates the entire profit to a few players, those in the minimum cut, while ignoring the rest; for an extreme example, see Example \ref{ex.fair-flow}, which has a core imputation giving the entire profit to a single agent and nothing to the rest.

Are there other core imputations for the max-flow game which spread the profit more evenly? Our first attempt at resolving this issue is to give an LP-based characterization of the core of the max-flow game; it yields core imputations, corresponding to fractional minimum cuts, which do spread the profit more evenly. A better resolution again involves the notions of min-max fair, max-min fair and equitable  core imputations.   

In several applications of the max-flow game, in particular network-related ones, individual agents build the ``pipes'' on which flow is sent, e.g., pipes carrying bandwidth on the Internet. Therefore, the question of distributing, in a fair way, the profits generated by sending the flow arises naturally. The transferable-utility feature of this game leads to the core as a preferred solution concept. Our attempts at finding ``more fair'' core imputations go a step further.

{\bf Encouraging diversity and avoiding over-representation via core imputations:} 
These issues arise in applications of the Hoffman-Kruskal game, which is a generalization of the assignment game, see Section \ref{sec.Hoffman}. This game is defined on a bipartite graph with edge weights. The objective is to find a maximum weight matching subject to vertex constraints (an upper bound on the number of times a vertex can be matched) and edge constraints (upper bound and lower bounds on the number of times an edge can be matched).

Let us describe two applications, among several, of the Hoffman-Kruskal game. The first is matching students to schools, with students and schools forming the bipartition, say $(U, V)$. Each vertex in $U$ represent a distinct category of students, e.g., minority students, women students with high grades and white male students with average grades, and its upper bound is the number of students in that category. Each vertex in $V$ represent a school, and its upper bound is the number of seats available in the school. The second application is matching medical residents to hospitals, with vertices in $U$ representing categories of residents, divided according to their speciality, and vertices in $V$ representing hospitals. In these applications, the edge lower and upper bounds are used for enforcing two types of constraints:

\begin{enumerate}
	\item {\em Encouraging diversity:} These are the lower bound constraints on edges. They can be used for ensuring that a school has a sufficient number of minority or women students, or that a rural hospital has a sufficient number of doctors having a specific speciality. 
	\item {\em Avoiding over-representation:} These are the upper bound constraints on edges. They can be used for ensuring that certain classes of students or doctors are not over-representation, so as to obtain a ``balanced'' allocation. 
\end{enumerate}

For the Hoffman-Kruskal game, core imputations corresponding to optimal dual solutions accomplish the following: whenever the lower or the upper bound constraint of an edge is binding, the profit allocated to the two agents at the end-points of this edge endogenously adjusts up or down, thereby either encouraging or discouraging further matchings of this edge, see Section \ref{sec.Hoffman} for the exact mechanics.

\subsection{Basic Gaps Remaining in the Use of LP-Duality Theory} 
\label{sec.gaps}

The classic paper of Shapley and Shubik \cite{Shapley1971assignment} gave a complete characterization of core imputations of the assignment game, namely they are  optimal solutions to the dual of the LP-relaxation of the maximum weight matching problem in the underlying graph. At the heart of their proof lies the fact that the polytope defined by the constraints of this LP is integral, i.e., its vertices have integral coordinates, see Definition \ref{def.vertices-integral}. This raises the following question:

\begin{quote}
	{\em  What is the root cause of integrality of the polytope defined by the linear system of an LP?} 
\end{quote}

For the assignment game, the root cause is {\em total unimodularity (TUM)}, see Definition \ref{def.totally-unimodular}. 

Shapley-Shubik was a paradigm-setting work. Over the years, in its vein, many researchers explored aspects of the core using TUM and LP-duality theory, see Section \ref{sec.related}, e.g., the importance of TUM was also observed by Deng et al. \cite{Deng1999algorithms}; however, they exploited it only in a limited setting, see Section \ref{sec.related} for details. Indeed, despite the extensive nature of this follow-up work, basic gaps still remain. The purpose of this paper, and its sequel \cite{Va.investment}, is to use the following three building blocks from LP-duality theory to address these gaps:
\begin{enumerate}
	\item Total unimodularity (TUM). 
	\item Complementary slackness conditions and strict complementarity.
	\item Total dual integrality (TDI). This is explored in \cite{Va.investment}, which extends the scope of the notion of core beyond profit --- equivalently cost or utility --- sharing.   
\end{enumerate}
 These building blocks have much to do with the theory of efficient\footnote{By ``efficient'' we mean polynomial time computable.} algorithms. That connection is not accidental; in fact, it is very deliberate, since in this paper, we are seeking insights into the core that are supported by efficient algorithms. Section \ref{sec.BS.vs.SS} clarifies this point by drawing a contrast between Bondareva-Shapley and Shapley-Shubik.

TUM enables us to characterize the cores of several games. What happens when the constraint matrix is not totally unimodular? It turns out that the polytope defined by the linear system of the LP may still be integral --- the more general condition that ascertains this fact is {\em total dual integrality} of the linear system. Exploring TDI is listed as the {\em third task}  above and is done in \cite{Va.investment}.

\subsection{Our Plan for Filling these Gaps}
\label{sec.plan}

Our {\em first task} in this paper is to exploit total unimodularity in depth, thereby characterizing the cores of several natural matching-based games and the max-flow game. 
We start by studying a generalization of the assignment game to the $b$-matching game in which an integral function $b$ specifies a bound on the number of matched edges that can be incident on a vertex. We consider two cases: $b$ is a constant function and $b$ is arbitrary. We also define two types of characterizations of the core: complete and partial, see Definition \ref{def.char-complete/partial}; in the former (latter), optimal dual solutions completely (partially) capture core imputations. We prove that the cores of the two cases of the $b$-matching game exhibit these two characterizations, respectively. For the latter game, this leads to the tantalizing question of understanding the origins of core imputations that do not correspond to optimal dual solutions.     

We next consider a very general LP-formulation, due to Hoffman and Kruskal \cite{Hoffman2010integral}, given in Theorem \ref{thm.Hoffman}, in which the constraint matrix is totally unimodular and the underlying polytope has integral vertices. We define a natural game, the {\em Hoffman-Kruskal game}, which is the most general matching-based game we study; it generalizes the assignment and $b$-matching games, see Section \ref{sec.enforcing} for applications of this game. We show that optimal dual solutions of this game partially characterize core imputations. We next ask if there are sub-cases of this game, besides the uniform bipartite $b$-matching game, for which optimal dual solutions completely capture core imputations. We provide a negative answer to this question. 

In all the generalizations of the assignment game stated above, the constraint matrix has $0/1$ entries only. On the other hand, a totally unimodular matrix has entries from the set $\{0, 1, -1\}$. We next study the core of the max-flow game, whose constraint matrix has entries from the set $\{0, 1, -1\}$. The importance of this game and our results for it are described in Section \ref{sec.enforcing}. For other results on the core of the max-flow game game, see Section \ref{sec.related}. 

Shapley and Shubik \cite{Shapley1971assignment} had also obtained other insights into the core of the assignment game. Since the core of this game consists of all optimal dual solutions, it is a convex polytope, and they asked the question: What does this polytope ``look like?'' They proved that if it contains two or more imputations, then it contains two ``antipodal'' ones which are maximally far apart in the core, see Theorem \ref{thm.extreme}. These two imputations maximally favor one side and disfavor the other side of the bipartition\footnote{Much like the top and bottom elements in a lattice of stable matchings.}.

Continuing is this vein, we seek further insights into the core of the assignment game; this leads to our {\em second task}, of exploring complementarity. The following broad observation indicates why complementarity needs to play an important role: The worth of the assignment game is given by an optimal solution to the primal LP and its core imputations are given by optimal solutions to the dual LP. The fundamental fact connecting these two solutions is complementary slackness conditions. Yet, despite the passage of half a century since  \cite{Shapley1971assignment}, the implications of this fact were not explored. Complementarity helps answer these three basic questions: 

\begin{enumerate}
	\item Do core imputations spread the profit more-or-less evenly or do they restrict profit to certain well-chosen agents? If the latter, what characterizes these ``chosen'' agents? An answer to this question will be critical for defining max-min fair and equitable core imputations as well, see Section \ref{sec.fair-matching}.  
	\item  An edge $(i, j)$ in the underlying graph is called a {\em team}. By definition, under any core imputation, the sum of profits of $i$ and $j$ is at least the weight of this edge. For which teams is the sum strictly larger than the weight of the edge? 
	\item How do core imputations behave in the presence of degeneracy?
	\end{enumerate}

Our answer to the first question is that the core rewards only {\em essential} agents, namely those who are matched by {\em every} maximum weight matching, see Theorem \ref{thm.vertices}. As stated in Section \ref{sec.enforcing}, the notion of essential agents enables us to define max-min fair core imputations. Our answer to the second question is quite counter-intuitive: we show that a team $(i, j)$ gets overpaid by some core imputation if and only if it is so incompetent that it doesn't participate in any maximum weight matching! If so, we show that at least one of $i$ and $j$ must be essential. Thus $i$ and $j$ do play well with other players but not with each other. This is the reason the sum of profits of $i$ and $j$ exceeds the weight of edge $(i, j)$. 
	
An assignment game is said to be {\em degenerate} if the optimal assignment is not unique. Although Shapley and Shubik had mentioned this phenomenon, they brushed it away, claiming that ``in the most common case'' the optimal assignment will be unique, and if not, their suggestion was to perturb the edge weights to make the optimal assignment unique. However, this is far from satisfactory, since perturbing the weights destroys crucial information contained in the original instance and the outcome becomes a function of the vagaries of the randomness imposed on the instance. Our answer to the third question is that degeneracy treats teams and agents in totally different ways, see  Section \ref{sec.degeneracy}. Section \ref{sec.related} discusses past approaches to degeneracy. 

Whereas the core of the assignment game is always non-empty, that of the general graph matching game can be empty. Deng et al. \cite{Deng1999algorithms} showed that the core of this game is non-empty if and only if the weights of maximum weight integral and fractional matchings concur. In this paper, we call such games {\em concurrent games}. For concurrent games, optimal dual solutions to Balinski's LP completely characterize core imputations.

Next, we study the three questions, raised above, for concurrent games, Section \ref{sec.general}.  The answers obtained for the first two questions are weaker than those for the assignment game. The underlying reason is that the characterization of the vertices of the Balinski polytope, Theorem \ref{thm.Balinski}, for concurrent games, is weaker than the characterization of the vertices of the polytope defined by the assignment game LP, Theorem \ref{thm.int-assn-LP}; in the former, vertices are half-integral matchings and in the latter, they are integral matchings. The answer to third question is identical to that of the assignment game.

\subsection{The Tennis-Club Analogy for Matching-Based Games}
\label{sec.tennis-club}

The matching game forms one of the cornerstones of cooperative game theory. The matching game can also be viewed as a matching market in which utilities of the agents are stated in monetary terms and side payments are allowed, i.e., it is a {\em transferable utility (TU) market}. For an extensive coverage of these notions, see the book by Moulin \cite{Moulin2014cooperative}.    

The following setting, taken from \cite{Eriksson2001stable} and \cite{Biro2012computing}, vividly captures the issues underlying profit-sharing in an assignment game (and its generalizations). Suppose a coed tennis club has sets $U$ and $V$ of women and men players, respectively, who can participate in an upcoming mixed doubles tournament. Assume $|U| = m$ and $|V| = n$, where $m, n$ are arbitrary. Let $G = (U, V, E)$ be a bipartite graph whose vertices are the women and men players and an edge $(i, j)$ represents the fact that agents $i \in U$ and $j \in V$ are eligible to participate as a mixed doubles team in the tournament. Let $w$ be an edge-weight function for $G$, where $w_{i j} > 0$ represents the expected earnings if $i$ and $j$ do participate as a team in the tournament. The total worth of the game is the weight of a maximum weight matching in $G$.

Assume that the club picks such a matching for the tournament. The question is how to distribute the total profit among the agents --- strong players, weak players and unmatched players --- so that no subset of players feel they will be better off seceding and forming their own tennis club. We will use this setting to discuss the issues involved in the questions raised above. 

The tennis club analogy extends to generalizations of the assignment game as follows. For the $b$-matching game, let $K$ denote the maximum $b$-value of a vertex. We will assume that the tennis club needs to enter teams into $K$ mixed doubles tournaments. The $b$-value of a player gives an upper bound on the number of tournaments the player can play. As before, the tennis club needs to maximize its total earnings. The Hoffman-Kruskal game generalizes the $b$-matching game in that each team has an upper and lower bound on how many tournaments it can play. 

For the general graph matching game, we will assume that the tennis club has players of one gender only and any two players can form a doubles team, thereby allowing the underlying graph to be non-bipartite.

\subsection{The Role of Integrality: Bondareva-Shapley vs Shapley-Shubik}
\label{sec.BS.vs.SS}

The early works of Bondareva \cite{Bondareva1963some} and Shapley \cite{Shapley1965balanced} gave a necessary and sufficient condition for non-emptiness of the core of a game, namely that it be balanced (see Section \ref{sec.prelim} for a definition); the proof of this fact is based on Farkas' Lemma from LP-duality theory.
 
 At the outset, the Bondareva-Shapley Theorem had the potential of being extremely useful, since it addresses a key property which one would like to ascertain about a game, namely non-emptiness of its core. Yet, in hindsight, it had no useful algorithmic consequences. The reason is two-fold: first, the LP employed for proving this theorem involves exponentially many variables, one for each subset of players and secondly, this LP is not built around any combinatorial structure\footnote{In contrast, Edmonds' linear programming relaxation for non-bipartite matching, which also has exponentially many variables, one for each odd subset of vertices, is amenable to an efficient solution \cite{GLS}. We believe the reason is the underlying combinatorial structure of matching.}. 
  
Shapley and Shubik \cite{Shapley1971assignment} introduced the ``right'' way of exploiting the power of LP-duality theory for characterizing core imputations. Perhaps their key ingenuity lay in choosing a game that not only had wide applicability and efficient solvability, but was also amenable to such an analysis. We discus the last point below; efficient solvability follows from the fact that the primal and dual LPs of the assignment game are small, i.e., they have polynomially many variables and constraints. 

An optimal solution to the LP-relaxation of the maximum weight matching problem is not ``obliged'' to be an integral matching in the underlying graph -- in general it will be a fractional matching. However, fractionally matching agents is not meaningful for the purpose of the assignment game. The saving grace is a property of the assignment game, namely that {\em its LP-relaxation always has an integral optimal solution}. A close examination of the Shapley-Shubik theorem reveals that their proof hinges precisely on this fact. In turn, the reason for this fact is that the polytope defined by the constraints of the primal LP has integral vertices; in this case, they are matchings in the graph. 

{\bf Organization of the paper:} Section \ref{sec.related} lists various works which explore the core via LP-duality theory. Section \ref{sec.prelim} provides basic definitions related to the core, as well as the setup of the assignment game in order to state in detail the theorems of Shapley and Shubik. In Section \ref{sec.Complementarity} we start with the second task, of viewing core imputations of the assignment game via the lens of complementarity; these results will also help provide insights into the new games studied in this paper as well as in obtaining the ``right'' definition of max-min fair core imputations in Section \ref{sec.fair}. Section \ref{sec.b-matching-game} defines and characterizes the core of the $b$-matching game and the uniform $b$-matching game. Section \ref{sec.Hoffman} does the same for their generalization to the Hoffman-Kruskal game. In Section \ref{sec.general} we define concurrent games and view their core imputations via the lens of complementarity. Section \ref{sec.flow} studies the core of the max-flow game via the lens of LP-duality. Finally in Section \ref{sec.fair} we define, and give ways of efficiently finding, min-max fair, max-min fair and equitable core imputations for all games studied in this paper.

\section{Related Works}
\label{sec.related}

We start by giving the result of Deng et al. \cite{Deng1999algorithms}. Let the set of agents of the game be denoted by $T = \{1, \cdots, n\}$ and let $w \in \R^m_+$ be an $m$-dimensional non-negative real vector specifying the weights of $m$ objects; in the assignment game, the objects are edges of the underlying bipartite graph. Let $A$ be an $n \times m$ matrix with $0/1$ entries whose $i^{th}$ row corresponds to agent $i \in T$. Let $x$ be an $m$-dimensional vector of variables and $\mathbb{1}$ be the $n$-dimensional vector of all 1s. Assume that the worth of the game is given by the objective function value of following integer program. 

	\begin{maxi}
		{} {w \cdot x}
			{\label{eq.IP}}
		{}
		\addConstraint{}{Ax \leq \mathbb{1}}
		\addConstraint{}{x \in \{0, 1\}}
	\end{maxi}

The worth of a sub-coalition, $T' \subseteq T$ is given by the integer program obtained by replacing $A$ by $A'$ in (\ref{eq.IP}), where $A'$ picks the set of rows corresponding to agents in $T'$. The LP-relaxation of (\ref{eq.IP}) is:

	\begin{maxi}
		{} {w \cdot x}
			{\label{eq.Primal}}
		{}
		\addConstraint{}{Ax \leq \mathbb{1}}
		\addConstraint{}{x \geq 0}
	\end{maxi}

Deng et al. proved that if LP (\ref{eq.Primal}) always has an integral optimal solution, then the set of core imputations of this game is exactly the set of optimal solutions to the dual of LP (\ref{eq.Primal}). As is well known, integrality holds if $A$ is totally unimodular, i.e., if every square sub-matrix of $A$ has determinant of 0, $+1$ or $-1$. This fact helps Deng et al.  characterize that the cores of several combinatorial optimization games, including maximum flow in unit capacity networks both directed and undirected, maximum number of edge-disjoint $s$-$t$ paths, maximum number of vertex-disjoint $s$-$t$ paths, maximum number of disjoint arborescences rooted at a vertex $r$, and concurrent games (defined below). 

A different kind of game, in which preferences are ordinal, is based on the stable matching problem defined by Gale and Shapley \cite{GaleS}. The only coalitions that matter in this game are ones formed by one agent from each side of the bipartition. A stable matching ensures that no such coalition has the incentive to secede and the set of such matchings constitute the core of this game. Vande Vate \cite{Vate1989linear} and Rothblum \cite{Rothblum1992characterization} gave linear programming formulations for stable matchings; the vertices of their underlying polytopes are integral and are stable matchings. More recently, Kiraly and Pap \cite{TDI-Kiraly2008total} showed that the linear system of Rothblum is in fact totally dual integral (TDI).

To deal with games having an empty core, e.g., the general graph matching game, the following two notions have been given in the past. The first is that of {\em least core}, defined by Mascher et al. \cite{Leastcore-Maschler1979geometric}. If the core is empty, there will necessarily be sets $S \subseteq V$ such that $v(S) < p(S)$ for any imputation $v$. The least core maximizes the minimum of $v(S) - p(S)$ over all sets $S \subseteq V$, subject to $v(\emptyset) = 0$ and $v(V) = p(V)$. A more well known notion is that of {\em nucleolus} which is contained in the least core.  After maximizing the minimum of $v(S) - p(S)$ over all sets $S \subseteq V$, it does the same for all remaining sets and so on. A formal definition is given below.

\begin{definition}
	\label{def.nucleolus}
	For an imputation $v: {V} \rightarrow \cR_+$, let $\theta(v)$ be the vector obtained by sorting the $2^{|V|} - 2$ values $v(S) - p(S)$ for each $\emptyset \subset S \subset V$ in non-decreasing order. Then the unique imputation, $v$, that lexicographically maximizes $\theta(v)$ is called the {\em nucleolus} and is denoted $\nu(G)$. 
\end{definition} 



In 1998, \cite{Faigle1998nucleon} stated the problem of computing the nucleolus of the matching game in polynomial time and Konemann et al. \cite{Konemann2020computing} found such an algorithm. However, their algorithm makes extensive use of the ellipsoid algorithm and is therefore neither efficient nor does it give deep insights into the underlying combinatorial structure. They leave the open problem of finding a combinatorial polynomial time algorithm. We note that the difference $v(S) - p(S)$ appearing in the least core and nucleolus has not been upper-bounded for any standard family of games, including the general graph matching game. 

A different notion was recently proposed in \cite{Va.general}, namely {\em approximate core}, and was used to obtain an imputation in the $2/3$-approximate core for the general graph matching game. This imputation can be computed in polynomial time, again using the power of LP-duality theory. This result was extended to $b$-matching games in general graphs by Xiaoet al. \cite{b-matching-approximate}. 

Konemann et al. \cite{b-matching-Konemann} showed that computing the nucleolus of the bipartite $b$-matching game, in which an edge can be matched at most once, is NP-hard even for the case $b=3$ for all vertices. Biro et al. \cite{Biro2012computing} showed co-NP-hardness of the following question: given an imputation for a $b$-matching game, decide if it belongs to the core.

Fang et al. \cite{Fang2002computational} show that testing membership in the core of a flow game is co-NP-hard and Deng et al. \cite{Deng-flow2009finding} show that computing the nucleolus of a flow game is NP-hard. Several researchers have studied the simple flow game, in which all edge capacities are unit. Deng et al. \cite{Deng-flow2009finding} and Potters et al. \cite{Potters2006nucleolus} give a polynomial time algorithm for for computing the nucleolus of such games; however, they need to use the ellipsoid algorithm. In contrast, Kern and Paulusma \cite{Kern} gave a combinatorial algorithm.  

Granot and Huberman \cite{Granot1981minimum, Granot-2-1984core} showed that the core of the minimum cost spanning tree game is non-empty and gave an algorithm for finding an imputation in it. Koh and Sanita \cite{Laura-Sanita} settle the question of efficiently determining if a spanning tree game is submodular; the core of such games is always non-empty. Nagamochi et al. \cite{Nagamochi1997complexity} characterize non-emptyness of core for the minimum base game in a matroid; the minimum spanning tree game is a special case of this game.



Over the years, researchers have approached the phenomenon of degeneracy in the assignment game from directions that are different from ours. Nunez and Rafels \cite{Nunez-Dimension}, studied relationships between degeneracy and the dimension of the core. They defined an agent to be {\em active} if her profit is not constant across the various imputations in the core, and non-active otherwise. Clearly, this notion has much to do with the dimension of the core, e.g., it is easy to see that if all agents are non-active, the core must be zero-dimensional. They prove that if all agents are active, then the core is full dimensional if and only if the game is non-degenerate. Furthermore, if there are exactly two optimal matchings, then the core can have any dimension between 1 and $m-1$, where $m$ is the smaller of $|U|$ and $|V|$; clearly, $m$ is an upper bound on the dimension.

In another work, Chambers and Echenique \cite{Chambers2015core} study the following question: Given the entire set of optimal matchings of a game on $m = |U|$, $n = |V|$ agents, is there an $m \times n$ surplus matrix which has this set of optimal matchings. They give necessary and sufficient conditions for the existence of such a matrix.

\section{Definitions and Preliminary Facts}
\label{sec.prelim}

\begin{definition}
	\label{def.cooperative-game}
	A {\em cooperative game} consists of a pair $(N, c)$ where $N$ is a set of $n$ agents and $v$ is the {\em characteristic function}; $c: 2^N \rightarrow \cR_+$, where for $S \subseteq N, \ c(S)$ is the {\em worth} that the sub-coalition $S$ can generate by itself. $N$ is also called the {\em grand coalition}.
\end{definition}

\begin{definition}
	\label{def.imputation}	
	An {\em imputation}
	is a function $p: N \rightarrow \cR_+$ that gives a way of dividing the worth of the game, $c(N)$, among the agents. It satisfies $\sum_{i \in N} {p(i)} = c(N)$; $p(i)$ is called the {\em profit} of agent $i$.  
\end{definition}


\begin{definition}
	\label{def.core}
	An imputation $p$ is said to be in the {\em core of the game} $(N, c)$ if for any sub-coalition $S \subseteq N$, the total profit allocated to agents in $S$ is at least as large as the worth that they can generate by themselves, i.e., $\sum_{i \in S} {p(i)} \geq c(S)$.
\end{definition}
 
\begin{definition}
	\label{def.balanced}
	A cooperative game $(N, c)$ is said to be {\em balanced} if for every function $\lambda: 2^N - \{\emptyset\} \rightarrow [0, 1]$ such that $\forall i \in N: \ \sum_{S \ni i} {\lambda(S)} = 1$, the following holds: $\sum_{2^N - \{\emptyset\}} {\lambda(S) \cdot c(S)} \leq c(N)$. 
\end{definition}

\begin{theorem}
	\label{thm.BS}
	(Bondareva-Shapley Theorem \cite{Bondareva1963some, Shapley1965balanced}) 
	A cooperative game $(N, c)$ has a non-empty core if and only if it is balanced. 
\end{theorem}

\begin{definition}
	\label{def.vertices-integral}
	We will say that a {\em polytope is integral} if its vertices have all integral coordinates. 
\end{definition}

\begin{definition}
	\label{def.totally-unimodular}
	Let $A$ be an $m \times n$ matrix with entries from the set $\{0, 1, -1\}$.  $A$ is said to be {\em totally unimodular\footnote{There is a polynomial time algorithm for checking if a matrix is totally unimodular \cite{Seymour1980decomposition}.} (TUM)} if every submatrix of $A$ has determinant $0, 1$ or $-1$. 
\end{definition}

\begin{theorem}
\label{thm.Hoffman}
[Hoffman and Kruskal \cite{Hoffman2010integral}] 
Let $A$ be an $n \times m$ with entries from the set $\{0, 1, -1\}$. Then $A$ is totally unimodular matrix if and only if for all integral vectors $a, b \in \ZZ_+^n$ and $c, d \in \ZZ^m$, the polyhedron
\[ x \in \R^m \ \ s.t. \ \ a \leq Ax \leq b, \ \ c \leq x \leq d \]
is integral.
\end{theorem}

\subsection{The Core of the Assignment Game}
\label{sec.core-assn-game}

The {\em assignment game}, $G = (U, V, E), \ w: E \rightarrow \cR_+$, has been defined in the Introduction. We start by giving definitions needed to state the Shapley-Shubik Theorem. For convenience, we will modify some of the notation that has been defined above to make it more suitable for the assignment game.  

A sub-coalition $(S_u \cup S_v)$ consists of a subset of the players with $S_u \subseteq U$ and $S_v \subseteq V$. The {\em worth} of a sub-coalition $(S_u \cup S_v)$ is defined to be the weight of a maximum weight matching in the graph $G$ restricted to vertices in $(S_u \cup S_v)$ and is denoted by $c(S_u \cup S_v)$. The {\em characteristic function} of the game is defined to be $c: 2^{U \cup V} \rightarrow \cR_+$. 

An {\em imputation}	consists of two functions $u: {U} \rightarrow \cR_+$ and $v: {V} \rightarrow \cR_+$ such that $\sum_{i \in U} {u_i} + \sum_{j \in V} {v_j} = c(U \cup V)$. Imputation $(u, v)$ is said to be in the {\em core of the assignment game} if for any sub-coalition $(S_u \cup S_v)$, the total profit allocated to agents in the sub-coalition is at least as large as the worth that they can generate by themselves, i.e., $\sum_{i \in S_u} {u_i} +  \sum_{j \in S_v} {v_j} \geq c(S)$. We next describe the characterization of the core of the assignment game given by Shapley and Shubik \cite{Shapley1971assignment}\footnote{Shapley and Shubik had described this game in the context of the housing market and had reduced it to the assignment game; the latter is a simpler setting.}

Linear program (\ref{eq.core-primal-bipartite}) gives the LP-relaxation of the problem of finding a maximum weight matching in $G$. In this program, variable $x_{ij}$ indicates the extent to which edge $(i, j)$ is picked in the solution. 

	\begin{maxi}
		{} {\sum_{(i, j) \in E}  {w_{ij} x_{ij}}}
			{\label{eq.core-primal-bipartite}}
		{}
		\addConstraint{\sum_{(i, j) \in E} {x_{ij}}}{\leq 1 \quad}{\forall i \in U}
		\addConstraint{\sum_{(i, j) \in E} {x_{ij}}}{\leq 1 }{\forall j \in V}
		\addConstraint{x_{ij}}{\geq 0}{\forall (i, j) \in E}
	\end{maxi}

Taking $u_i$ and $v_j$ to be the dual variables for the first and second constraints of (\ref{eq.core-primal-bipartite}), we obtain the dual LP: 

 	\begin{mini}
		{} {\sum_{i \in U}  {u_{i}} + \sum_{j \in V} {v_j}} 
			{\label{eq.core-dual-bipartite}}
		{}
		\addConstraint{ u_i + v_j}{ \geq w_{ij} \quad }{\forall (i, j) \in E}
		\addConstraint{u_{i}}{\geq 0}{\forall i \in U}
		\addConstraint{v_{j}}{\geq 0}{\forall j \in V}
	\end{mini}

It is easy to see that the constraint matrix of LP (\ref{eq.core-primal-bipartite}) is totally unimodular and this LP is clearly a special case of the general formulation given in Theorem \ref{thm.Hoffman}. The next theorem follows from this fact; for a proof see \cite{LP.book}.

\begin{theorem}
	\label{thm.int-assn-LP}
The polytope defined by the constraints of LP (\ref{eq.core-primal-bipartite}) is integral; its vertices are matchings in the underlying graph. 
\end{theorem}

The proof of the next theorem hinges on the key fact stated in Theorem \ref{thm.int-assn-LP}.

\begin{theorem}
	\label{thm.SS}
	(Shapley and Shubik \cite{Shapley1971assignment})
	The imputation $(u, v)$ is in the core of the assignment game if and only if it is an optimal solution to the dual LP, (\ref{eq.core-dual-bipartite}). 
\end{theorem}
	
\begin{definition}
	\label{def.char-complete/partial}
Given a cooperative game, we will say that {\em the dual completely characterizes its core} if there is a one-to-one correspondence between core imputations and optimal solutions to the dual. On the other hand, if every optimal solution to the dual corresponds to a core imputation but there are core imputations that do not correspond to optimal solutions to the dual, then we will say that {\em the dual partially characterizes its core}. 
\end{definition}

By Theorem \ref{thm.SS}, the assignment game's dual completely characterizes its core. Clearly, the core is a convex polyhedron. Shapley and Shubik had shed further light on the structure of the core by showing that it has two special imputations which are furtherest apart and so can be thought of as antipodal imputations. In the tennis club setup, one of these imputations maximizes the earnings of women players and minimizes the earnings of men players, and the second does exactly the opposite.  

For $i \in U$, let $u_i^h$ and $u_i^l$ denote the highest and lowest profits that $i$ accrues among all imputations in the core. Similarly, for $j \in V$, let $v_j^h$ and $v_j^l$ denote the highest and lowest profits that $j$ accrues in the core. Let $u^h$ and $u^l$ denote the vectors whose components are $u_i^h$ and $u_i^l$, respectively. Similarly, let $v^h$ and $v^l$ denote vectors whose components are $v_j^h$ and $v_j^l$, respectively. The following is a formal statement regarding the extreme imputations.

\begin{theorem}
	\label{thm.extreme}
		(Shapley and Shubik \cite{Shapley1971assignment})
The core of the assignment game has two extreme imputations; they are $(u^h, v^l)$ and $(u^l, v^h)$.
\end{theorem} 

\section{Complementarity Applied to the Assignment Game}
\label{sec.Complementarity}

In this section, we provide answers to the three questions, for assignment games, which were 
 raised in the Introduction.

\subsection{The first question: Allocations made to agents by core imputations}
\label{sec.vertices}

\begin{definition}
	\label{def.player}
	A generic player in $U \cup V$ will be denoted by $q$. We will say that $q$ is:
	\begin{enumerate}
		\item {\em essential} if $q$ is matched in every maximum weight matching in $G$.
		\item {\em viable} if there is a maximum weight matching $M$ such that $q$ is matched in $M$ and another, $M'$ such that $q$ is not matched in $M'$. 	
		\item {\em subpar} if for every maximum weight matching $M$ in $G$, $q$ is not matched in $M$. 	
		\end{enumerate}
\end{definition}

\begin{definition}
\label{def.player-paid}
	Let $y$ be an imputation in the core. We will say that $q$ {\em gets paid in $y$} if $y_q > 0$ and {\em does not get paid} otherwise. Furthermore, $q$ is {\em paid sometimes} if there is at least one imputation in the core under which $q$ gets paid, and it is {\em never paid} if it is not paid under every imputation. 
\end{definition}

\begin{theorem}
	\label{thm.vertices}
	 For every player $q \in (U \cup V)$: 
		\[ q \ \mbox{is paid sometimes}  \ \iff \ q \ \mbox{is essential} \]  
\end{theorem}
	
\begin{proof}
The proof follows by applying complementary slackness conditions and strict complementarity to the primal LP (\ref{eq.core-primal-bipartite}) and dual LP (\ref{eq.core-dual-bipartite}); see \cite{Sch-book} for formal statements of these facts. By Theorem \ref{thm.SS}, talking about imputations in the core of the assignment game is equivalent to talking about optimal solutions to the dual LP.

 Let $x$ and $y$ be optimal solutions to LP (\ref{eq.core-primal-bipartite}) and LP (\ref{eq.core-dual-bipartite}), respectively. By the Complementary Slackness Theorem, for each $q \in (U \cup V): \ y_q \cdot (x(\delta(q)) - 1) = 0$. 

$(\Rightarrow)$  Suppose $q$ is paid sometimes. Then, there is an optimal solution to the dual LP, say $y$, such that $y_q > 0$. By the Complementary Slackness Theorem, for any optimal solution, $x$, to LP (\ref{eq.core-primal-bipartite}), $x(\delta(q)) = 1$, i.e., $q$ is matched in $x$. Varying $x$ over all optimal primal solutions, we get that $q$ is always matched. In particular, $q$ is matched in all optimal assignments, i.e., integral optimal primal solutions, and is therefore essential. This proves the forward direction.

$(\Leftarrow)$ Strict complementarity implies that corresponding to each player $q$, there is a pair of optimal primal and dual solutions, say $x$ and $y$, such that either $y_q = 0$ or $x(\delta(q)) = 1$ but not both. Assume that $q$ is essential, i.e., it is matched in every integral optimal primal solution. 
 
 We will use Theorem \ref{thm.int-assn-LP}, which implies that every fractional optimal primal solution to LP (\ref{eq.core-primal-bipartite}) is a convex combination of integral optimal primal solutions. Therefore $q$ is fully matched in every optimal solution, $x$, to LP (\ref{eq.core-primal-bipartite}), i.e., $x(\delta(q)) = 1$, so there must be an optimal dual solution $y$ such that $y_q > 0$. Hence $q$ is paid sometimes, proving the reverse direction. 
\end{proof}

Theorem \ref{thm.vertices} is equivalent to the following. For every player $q \in (U \cup V)$: 
		\[ q \ \mbox{is never paid} \ \iff \ q \ \mbox{is not essential} \]  
		
Thus core imputations pay only essential players and each of them is paid in some core imputation.  Since we have assumed that the weight of each edge is positive, so is the worth of the game, and all of it goes to essential players. Hence we get:

\begin{corollary}
	\label{cor.vertices}
	In the assignment game, the set of essential players is non-empty and in every core imputation, the entire worth of the game is distributed among essential players; moreover, each of them is paid in some core imputation. 
\end{corollary} 

By Corollary \ref{cor.vertices}, core imputations reward only essential players. This raises the following question: Can't a non-essential player, say $q$, team up with another player, say $p$, and secede, by promising $p$ almost all of the resulting profit? The answer is ``No'', because the dual (\ref{eq.core-dual-bipartite}) has the constraint $y_q + y_p \geq w_{qp}$. Therefore, if $y_q = 0$, $y_p \geq w_{q p}$, i.e., $p$ will not gain by seceding together with $q$.

\subsection{The second question: Allocations made to teams by core imputations}
\label{sec.edges}

\begin{definition}
	\label{def.team}
	By a {\em team} we mean an edge in $G$; a generic one will be denoted as $e = (u, v)$. We will say that $e$ is:
	\begin{enumerate}
		\item {\em essential} if $e$ is matched in every maximum weight matching in $G$.
		\item {\em viable} if there is a maximum weight matching $M$ such that $e \in M$, and another, $M'$ such that $e \notin M'$. 
		\item {\em subpar} if for every maximum weight matching $M$ in $G$, $e \notin M$. 
	\end{enumerate}
	\end{definition}
	
\begin{definition}
\label{def.team-paid}
	 Let $y$ be an imputation in the core of the game. We will say that $e$ is {\em fairly paid in $y$} if $y_u + y_v = w_e$ and it is {\em overpaid} if $y_u + y_v > w_e$\footnote{Observe that by the first constraint of the dual LP (\ref{eq.core-dual-bipartite}), these are the only possibilities.}. Finally, we will say that $e$ is {\em always paid fairly} if it is fairly paid in every imputation in the core.
\end{definition}

\begin{theorem}
	\label{thm.edges}
	 For every team $e \in E$: 
		\[ e \ \mbox{is always paid fairly} \ \iff \ e \ \mbox{is viable or essential} \]
\end{theorem}
	
\begin{proof}
The proof is similar to that of Theorem \ref{thm.vertices}. Let $x$ and $y$ be optimal solutions to LP (\ref{eq.core-primal-bipartite}) and LP (\ref{eq.core-dual-bipartite}), respectively. By the Complementary Slackness Theorem, for each $e = (u, v) \in E: \ \ x_e \cdot (y_u + y_v - w_e) = 0$.

$(\Leftarrow)$ To prove the reverse direction, suppose $e$ is viable or essential. Then there is an optimal solution to the primal, say $x$, under which it is matched. Therefore,  $x_e > 0$. Let $y$ be an arbitrary optimal dual solution. Then, by the Complementary Slackness Theorem, $y_u + y_v = w_e$, i.e., $e$ is fairly paid in $y$. Varying $y$ over all optimal dual solutions, we get that $e$ is always paid fairly. 

$(\Rightarrow)$ To prove the forward direction, we will use strict complementarity. It implies that corresponding to each team $e$, there is a pair of optimal primal and dual solutions $x$ and $y$ such that either $x_e = 0$ or $y_u + y_v = w_e$ but not both. 

Assume that team $e$ is always fairly paid, i.e., under every optimal dual solution $y$, $y_u + y_v = w_e$. By strict complementarity, there must be an optimal primal solution $x$ for which $x_e > 0$. Theorem \ref{thm.int-assn-LP} implies that $x$ is a convex combination of optimal assignments. Therefore, there must be an optimal assignment in which $e$ is matched. Therefore $e$ is viable or essential and the forward direction also holds. 
\end{proof}

\begin{corollary}
	\label{cor.endpoint-essential}
	Let $e = (u, v)$ be a subpar team. Then at least one of $u$ and $v$ is essential. 
\end{corollary}

\begin{proof}
In every maximum weight matching, at least one of $u$ and $v$ must be matched, since otherwise $e$ should get matched. By Theorem \ref{thm.edges}, there is a core imputation under which $e$ is overpaid. Therefore, under this imputation, at least one of $u$ and $v$ is paid and by Theorem \ref{thm.vertices}, that vertex must be essential. 
\end{proof}

Negating both sides of the implication proved in Theorem \ref{thm.edges}, we get the following  implication. For every team $e \in E$: 
		\[ e \ \mbox{is subpar} \ \iff \ e \ \mbox{is sometimes overpaid} \]

Clearly, this statement is equivalent to the statement proved Theorem \ref{thm.edges} and hence contains no new information. However, it provides a new viewpoint. These two equivalent  statements yield the following assertion, which at first sight seems incongruous with what we desire from the notion of the core and the just manner in which it allocates profits:

\begin{center}
{\em Whereas viable and essential teams are always paid fairly, subpar teams are sometimes overpaid.}
\end{center}

How can the core favor subpar teams over viable and essential teams? An explanation is provided in the Introduction, namely a subpar team $(i, j)$ gets overpaid because $i$ and $j$ create worth by playing in competent teams with other players. Finally, we observe that contrary to Corollary \ref{cor.vertices}, which says that the set of essential players is non-empty, it is easy to construct examples in which the set of essential teams may be empty. 

By Theorem \ref{thm.vertices}, different essential players are paid in different core imputations and by Theorem \ref{thm.edges}, different subpar teams are overpaid in different core imputations. This raises the following question: is there a core imputation that simultaneously satisfies all these conditions? Theorem \ref{thm.simultaneous} gives a positive answer. 

\begin{theorem}
	\label{thm.simultaneous}
	For the assignment game, there is a core imputation satisfying:
\begin{enumerate}
	\item a player $q \in U \cup V$ gets paid if and only if $q$ is essential.
	\item a team $e \in E$ gets overpaid if and only if $e$ is subpar. 
\end{enumerate}
\end{theorem}

\begin{proof}
By Theorem \ref{thm.vertices}, for each essential player $q$, there is a core imputation under which $q$ gets paid and by Theorem \ref{thm.edges}, for each subpar team $e$, there is a core imputation under which $e$ gets overpaid. Consider a convex combination of all these imputations; it must give positive weight to each of these imputations. Clearly, this is a core imputation. 

We observe that none of the core imputations pay non-essential players or overpay non-subpar teams. Consequently, the imputation constructed above satisfies the conditions of the theorem. 
\end{proof}

\subsection{The third question: Degeneracy}
\label{sec.degeneracy}

 Next we use Theorems \ref{thm.vertices} and \ref{thm.edges} to get insights into degeneracy. Clearly, if an assignment game is non-degenerate, then every team and every player is either always matched or always unmatched in the set of maximum weight matchings in $G$, i.e., there are no viable teams or players. Since viable teams and players arise due to degeneracy, in order to understand the phenomenon of degeneracy, we need to understand how viable teams and players behave with respect to core imputations; this is done in the next corollary.
 
\begin{corollary}
	\label{cor.degen}
	In the presence of degeneracy, imputations in the core of an assignment game treat:
	\begin{itemize}
			\item  viable players in the same way as subpar players, namely they are never paid.  
		\item viable teams in the same way as essential teams, namely they are always fairly paid. 
	\end{itemize}
\end{corollary}

\section{The Bipartite $b$-Matching Game}
\label{sec.b-matching-game}

In this section, we will define the bipartite $b$-matching game and its special case when $b$ is the constant function and we will study their core imputations; both versions generalize the assignment game.

\subsection{Definitions and Preliminary Facts}
\label{sec.b-prelim}

As in the assignment game, let $G = (U, V, E), \ w: E \rightarrow \cR_+$ be the underlying bipartite graph and edge-weight function. Let function $b: U \cup V \rightarrow \ZZ_+$ give an upper bound on the number of times a vertex can be matched. An edge can be matched multiple number of times; however, limits imposed by $b$ on vertices will impose limits on edges. Thus edge $(i, j)$ can be matched at most $\min \{b_i, b_j\}$ times. Any choice of edges, with multiplicity, subject to these constraints, is called a $b$-matching. 

Under the {\em bipartite $b$-matching game}, the {\em worth} of a coalition $(S_u \cup S_v)$, with $S_u \subseteq U, S_v \subseteq V$, is the weight of a maximum weight $b$-matching in the graph $G$ restricted to vertices in $(S_u \cup S_v)$ only. We will denote this by $c(S_u \cup S_v)$; the {\em characteristic function} of the game is defined to be $c: 2^{U \cup V} \rightarrow \cR_+$. An {\em imputation}	consists of two functions $\alpha : {U} \rightarrow \cR_+$ and $\beta : {V} \rightarrow \cR_+$ such that $\sum_{i \in U} {\alpha (i)} + \sum_{j \in V} {\beta (j)} = c(U \cup V)$. Definition \ref{def.core}, defining the core, carries over unchanged. The special case of the bipartite $b$-matching game in which $b$ is the constant function is called  the {\em uniform bipartite $b$-matching game}; we will denote the constant by $b_c \in \ZZ_+$. As stated in the Introduction, the tennis club analogy applies to these two games as well. 

Linear program (\ref{eq.b-uncon-core-primal-bipartite}) gives the LP-relaxation of the problem of finding a maximum weight $b$-matching. In this program, variable $x_{ij}$ indicates the extent to which edge $(i, j)$ is picked in the solution; observe that  there is no upper bound on the variables $x_{ij}$ since an edge can be matched any number of times.

	\begin{maxi}
		{} {\sum_{(i, j) \in E}  {w_{ij} x_{ij}}}
			{\label{eq.b-uncon-core-primal-bipartite}}
		{}
		\addConstraint{\sum_{(i, j) \in E} {x_{ij}}}{\leq b_i \quad}{\forall i \in U}
		\addConstraint{\sum_{(i, j) \in E} {x_{ij}}}{\leq b_j }{\forall j \in V}
		\addConstraint{x_{ij}}{\geq 0}{\forall (i, j) \in E}
	\end{maxi}

Taking $u_i$ and $v_j$ to be the dual variables for the first and second constraints of (\ref{eq.b-uncon-core-primal-bipartite}), we obtain the dual LP: 

 	\begin{mini}
		{} {\sum_{i \in U}  {b_i u_{i}} + \sum_{j \in V} {b_j v_j}} 
			{\label{eq.b-uncon-core-dual-bipartite}}
		{}
		\addConstraint{ u_i + v_j}{ \geq w_{ij} \quad }{\forall (i, j) \in E}
		\addConstraint{u_{i}}{\geq 0}{\forall i \in U}
		\addConstraint{v_{j}}{\geq 0}{\forall j \in V}
	\end{mini}

	\bigskip
	
As in the assignment game, the constraint matrix of LP (\ref{eq.b-uncon-core-primal-bipartite}) is  totally unimodular and this LP is a special case of the general formulation given in Theorem \ref{thm.Hoffman}. Consequently, the underlying polyhedron has all integral vertices. Once again, this fact plays a key role in characterizing the core of the (uniform) bipartite $b$-matching game.

\subsection{The Core of the Uniform Bipartite $b$-Matching Game}
\label{sec.b-uniform-core}

The next theorem is analogous to the Shapley-Shubik Theorem.

\begin{theorem}
	\label{thm.b-uniform}
	For the uniform bipartite $b$-matching game, the dual completely characterizes its core.
\end{theorem}

\begin{proof}
The proof hinges on the fact that the polytope defined by the constraints of the primal LP, (\ref{eq.b-uncon-core-primal-bipartite}) has integral vertices, i.e., they are $b$-matchings in $G$. 

Let $(u, v)$ be an optimal dual solution. By integrality and the LP-duality theorem, the worth of the game, 
$$c(U \cup V) = b_c \cdot (\sum_{i \in U} {u_i} + \sum_{j \in V} {v_j}) .$$
Therefore $(\alpha, \beta)$ is an imputation, where $\alpha_i = b_c \cdot u_i$ and $\beta_j = b_c \cdot v_j$. 
	 
Consider a sub-coalition $(S_u \cup S_v)$, with $S_u \subseteq U, S_v \subseteq V$. Let $G'$ denote the subgraph of $G$ induced on the vertices $(S_u \cup S_v)$. Once again by integrality and the LP-duality theorem, we get that worth of $(S_u \cup S_v)$ equals the objective function value of  the optimal dual for graph  $G'$. Since the restriction of $(u, v)$ to $G'$ is a feasible dual for $G'$, we get that $b_c \cdot (\sum_{i \in S_u} {u_i} + \sum_{j \in S_v} {v_j}) =  \sum_{i \in S_u} {\alpha_i} + \sum_{j \in S_v} {\beta_j} \geq c(S_u \cup S_v)$, i.e., the core condition is satisfied for sub-coalition $(S_u \cup S_v)$. Therefore $(\alpha, \beta)$ is a core imputation.

Next, let $(\alpha, \beta)$ be a core imputation. By integrality and the LP-duality theorem,  
$$ \sum_{i \in U} {\alpha_i} + \sum_{j \in V} {\beta_j} \ = c(U \cup V) .$$
Let $u_i = {1 \over b_c} \alpha_i$ and $v_j = {1 \over b_c} \beta_j$. We will show that $(u, v)$ is an optimal dual solution for $G$, thereby proving the theorem.

Corresponding to any edge $e = (i, j)$, consider the sub-coalition $S = \{ i, j\}$. The worth of this sub-coalition is obtained by picking edge $e$ $b_c$ times, i.e., $c(\{i, j\}) = b_c \cdot w_e$.  Since $(\alpha, \beta)$ be a core imputation, the profit allocated to this sub-coalition is at least its worth, i.e., $\alpha_i + \beta_j \geq b_c \cdot w_e$. Dividing by $b_c$ we get $u_i + v_j \geq w_e$. Therefore $(u, v)$ satisfies the constraint in LP (\ref{eq.b-uncon-core-dual-bipartite}) and is hence a dual feasible solution. Since $(\alpha, \beta)$ is a core imputation, $ \sum_{i \in U} {\alpha_i} + \sum_{j \in V} {\beta_j} \ = c(U \cup V) .$ Therefore, by integrality and the LP-duality theorem, the objective function value of this dual equals the  optimal primal. Therefore $(u, v)$ is an optimal dual solution.  
\end{proof}

\begin{remark}
	\label{rem.simpler}
Clearly, the proof given above can be used for characterizing the core of the assignment game as well. We note that the proof given in \cite{Shapley1971assignment} does not explicitly use total unimodularity. Instead, it explicitly uses the maximum weight matching guaranteed by the TUM of LP (\ref{eq.core-primal-bipartite}). Our proof is simpler and more modular. Because of the latter, its idea applies directly to other games admitting TUM. We further observe that with a little bit of care, Theorem \ref{thm.b-uniform} could have been derived from Theorem \ref{thm.SS}. However, by doing that we would have lost the opportunity of stating this idea in a simple setting before applying it to more complex games. 
\end{remark}

Next, we prove that the core of the uniform bipartite $b$-matching game also has two extreme imputations, as claimed for the assignment game in Theorem \ref{thm.extreme}. For $i \in U$, let $\alpha_i^h$ and $\alpha_i^l$ denote the highest and lowest profits that $i$ accrues among all imputations in the core. Similarly, for $j \in V$, let $\beta_j^h$ and $\beta_j^l$ denote the highest and lowest profits that $j$ accrues in the core. Let $\alpha^h$ and $\alpha^l$ denote the vectors whose components are $\alpha_i^h$ and $\alpha_i^l$, respectively. Similarly, let $\beta^h$ and $\beta^l$ denote vectors whose components are $\beta_j^h$ and $\beta_j^l$, respectively. 

\begin{theorem}
	\label{thm.b-extreme}	
	The core of the uniform bipartite $b$-matching game has two extreme imputations; they are $(\alpha^h, \beta^l)$ and $(\alpha^l, \beta^h)$.
\end{theorem} 

The proof of this theorem follows from Lemma \ref{lem.extreme-lemma}, whose proof is straightforward and is omitted. Let $(q, r)$ and $(s, t)$ be two imputations in the core of the uniform bipartite $b$-matching game. For each $i \in U$, let 
\[ \un{\alpha_i} = \min(q_i, s_i) \ \ \ \  \mbox{and} \ \ \ \  \ov{\alpha_i} = \max (q_i, s_i) .\]
Further, for each $j \in R$, let
\[  \un{\beta_j} = \min (r_j , t_j) \ \ \ \  \mbox{and} \ \ \ \  \ov{\beta_j} = \max (r_j, t_j) .\] 

\begin{lemma}
	\label{lem.extreme-lemma}
	$(\un{\alpha} , \ov{\beta} )$ and $(\ov{\alpha}, \un{\beta} )$ are imputations in the core of the uniform bipartite $b$-matching game.
\end{lemma}

Finally, we use complementarity to gain further insights into core imputations of the uniform bipartite $b$-matching game.

\begin{definition}
	\label{def.b-player}
	A generic player in $U \cup V$ will be denoted by $q$. We will say that $q$ is:
	\begin{enumerate}
		\item {\em essential} if $q$ is matched $b$ times in every maximum weight matching in $G$.
		\item {\em viable} if there is a maximum weight matching $M$ such that $q$ is matched $b$ times in $M$ and another, $M'$ such that $q$ is matched less than $b$ times in $M'$. 	
		\item {\em subpar} if for every maximum weight matching $M$ in $G$, $q$ is matched than $b$ times in $M$. 	
		\end{enumerate}
\end{definition}
	
Definitions \ref{def.player-paid}, \ref{def.team} and \ref{def.team-paid} carry over verbatim. Under these definitions, it is easy to check that all theorems and corollaries stated in Section \ref{sec.Complementarity} for the assignment game carry over to the uniform bipartite $b$-matching game.

\subsection{The Core of the Bipartite $b$-Matching Game}
\label{sec.b-arbitrary-core}

We next study the case that $b$ is an arbitrary function. As stated in Theorem \ref{thm.b-arbitrary}, for this game the dual partially characterizes its core. Corresponding to each optimal solution $(u, v)$ to the dual LP (\ref{eq.b-uncon-core-dual-bipartite}) there exists a core imputation $(\alpha, \beta)$, where the profit allocated to $i \in U$ is $\alpha_i = b_i \cdot u_i$ and that to $j \in V$ is $\beta_j = b_j \cdot v_j$. The proof of this statement is exactly the same as the the analogous statement in Theorem \ref{thm.b-uniform} and is omitted. Example \ref{ex.b-arbitrary} gives an instance of a bipartite $b$-matching game which has core imputations that do not correspond to any optimal dual solution. Hence we get: 

\begin{theorem}
	\label{thm.b-arbitrary}
	For the bipartite $b$-matching game, the dual partially characterizes its core.
\end{theorem}

\begin{corollary}
\label{cor.b-arbitrary}
	The core of the bipartite $b$-matching game is always non-empty. 
\end{corollary}


\begin{figure}[h]
\begin{center}
\includegraphics[width=2.4in]{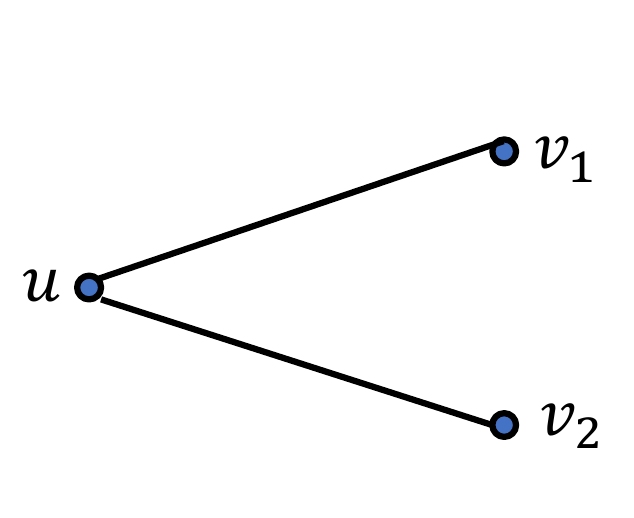}
\caption{The graph for Example \ref{ex.b-arbitrary}.}
\label{fig.7}
\end{center}
\end{figure}


\begin{example}
	\label{ex.b-arbitrary}
For the bipartite $b$-matching game defined by the graph of Figure \ref{fig.7}, let the $b$ values be $(2, 2, 1)$ for $(u, v_1, v_2)$, and let the edge weights be $1$ and $3$ for $(u, v_1)$ and $(u, v_2)$, respectively. 

By matching both edges once each we get that the worth of the game is 4. The unique optimal dual solution is $(1, 0, 2)$ for $(u, v_1, v_2)$. It is easy to see that allocation of profits of $(4, 0, 0)$ for $(u, v_1, v_2)$ is a core imputation. The corresponding dual solution would have been $(2, 0, 0)$; however, it is infeasible. Therefore this core imputation does not correspond to an optimal dual solution.
\end{example}

\begin{remark}
\label{rem.char}
Let $I$ be an instance of the bipartite $b$-matching game and let $D(I)$ denote the set of its core imputations which are optimal solutions to the dual LP. Since the optimal dual solutions don't capture all core imputations, the characterizations established in Theorems \ref{thm.vertices} and \ref{thm.edges} for the assignment game, don't carry over. Instead, we will restrict all definitions about payments to core imputations in $D(I)$ only. For players and teams we will use Definition \ref{def.b-player} and Definition \ref{def.team}, respectively. For payments to players and teams, Definitions \ref{def.player-paid} and \ref{def.team-paid} carry over, provided they are restricted to core imputations in $D(I)$ only. Under these definitions, it is easy to check that all theorems and corollaries stated in Section \ref{sec.Complementarity} for the assignment game carry over to the bipartite $b$-matching game as well. In particular, imputations in $D(I)$ pay  essential players only.  
\end{remark}

\section{The Core of the Hoffman-Kruskal Game}
\label{sec.Hoffman}

In this section, we will define a game that generalizes the $b$-matching game and whose LP-relaxation sacrifices as little generality as possible compared to the LP-formulation given in the theorem of Hoffman and Kruskal, Theorem \ref{thm.Hoffman}. The game we obtain is a natural one; its applications include matching students to schools and medical residents to hospitals. As described below, its dual variables provide a way of enforcing constraints that arise in such applications. The game has an interesting interpretation in the tennis club setting as well, as described in Section \ref{sec.tennis-club}.

The {\em Hoffman-Kruskal game} is defined as follows. Let $G = (U, V, E), \ w: E \rightarrow \cR_+$ be a bipartite graph and edge-weight function. The set of {\em agents} is $U \cup V$ and the set of {\em teams} is $E$. Let function $b: U \cup V \rightarrow \ZZ_+$ give an upper bound on the number of times a vertex can be matched. Further, let functions $c: E \rightarrow \ZZ_+$ and $d: E\rightarrow \ZZ_+$ give a lower bound and an upper bound, respectively, on the number of times an edge can be matched. A matching obeying these conditions will be called an {\em HK matching}. 

Linear program  (\ref{eq.Hoff-primal}) gives the LP-relaxation of the problem of finding a maximum weight HK matching.

		\begin{maxi}
		{} {\sum_{(i, j) \in E}  {w_{ij} x_{ij}}}
			{\label{eq.Hoff-primal}}
		{}
		\addConstraint{\sum_{(i, j) \in E} {x_{ij}}}{\leq b_i \quad}{\forall i \in U}
		\addConstraint{\sum_{(i, j) \in E} {x_{ij}}}{\leq b_j \quad}{\forall j \in V}
		\addConstraint{x_{ij}}{\geq c_{ij}}{\forall (i, j) \in E}
		\addConstraint{x_{ij}}{\leq d_{ij}}{\forall (i, j) \in E}
	\end{maxi}

Taking $u_i$, $v_j$, $y_{ij}$ and $z_{ij}$  to be the dual variables for the first to the fourth constraints, respectively, of (\ref{eq.Hoff-primal}), we obtain the dual LP: 

 	\begin{mini}
		{} {\sum_{i \in U}  {b_i u_i} + \sum_{j \in V} {b_j v_j} + \sum_{(i, j) \in E}  {(d_{ij} z_{ij} - c_{ij} y_{ij})}} 
			{\label{eq.Hoff-dual}}
		{}
		\addConstraint{u_i + v_j + z_{ij} - y_{ij}}{\geq w_{ij} \quad }{\forall (i, j) \in E}
		\addConstraint{u_{i}}{\geq 0}{\forall i \in U}
		\addConstraint{v_j}{\geq 0}{\forall j \in V}
		\addConstraint{y_{ij},z_{ij}}{\geq 0}{\forall (i, j) \in E}
	\end{mini}

We next state the manner in which the LP-relaxation for the Hoffman-Kruskal game is less general as compared to  the LP in Theorem \ref{thm.Hoffman}. For this purpose, recall the latter LP: 
\[ x \in \R^m \ \ s.t. \ \ a \leq Ax \leq b, \ \ c \leq x \leq d \]
\begin{enumerate}
	\item The constraint matrix has entries from $\{0, 1\}$ rather than $\{0, 1, -1\}$. In Section \ref{sec.flow}, we give a partial characterization of the max-flow game; its constraint matrix has $-1$ entries as well. 
	\item The lower bound on the number of times a vertex can be matched, $a \leq Ax$, has been removed. The reason is the following: if this constraint were added, then the form of the dual LP dictates that agents may receive negative profit. This would violate individual rationality and of course would be a reason for the agent to secede, hence going against the spirit of the core. 
	\item The vectors are constrained to have non-negative integral entries rather than arbitrary integral entries, since constraining an agent to play a negative number of games is meaningless.  
\end{enumerate}

The constraints on edges make the Hoffman-Kruskal game considerably more complex than the $b$-matching game; in particular, the dual contains not only vertex-variables, which would help us characterize the core of this game, but also edge-variables. However, this is a double-edged sword in that the latter variables help enforce constraints that arise naturally in applications of this game. 

Let us describe two applications among several. The first is matching students to schools. One way to use the Hoffman-Kruskal game in this setting is to let each vertex in $U$ represent a distinct category of students, e.g., minority students, women students with high grades, white male students with average grades etc. For $i \in U$, $b_i$ is the number of students in category $i$. Let vertices in $V$ represent schools, with $b_j$ for $j \in V$ being the number of students which school $j$ can admit. The second application is matching medical residents to hospitals, in which vertices in $U$ represent  categories of residents, divided according to their speciality, and $V$ represents hospitals. The constraints which the Hoffman-Kruskal game can handle are of two types:

\begin{enumerate}
	\item {\em Diversity:} These are the lower bound constraints on edges. They can be used for ensuring that a school has a sufficient number of minority or women students, or that a rural hospital has a sufficient number of doctors having a specific speciality. 
	\item {\em Avoiding over-representation:} These are the upper bound constraints on edges. They can be used for ensuring that certain classes of students or doctors are not over-representation, so as to obtain a ``balanced'' allocation. 
\end{enumerate}

Let us start by rewriting the constraint corresponding to edge $(i, j)$ in the dual LP (\ref{eq.Hoff-dual}) as follows: 

$$ u_i + v_j \geq w_{ij} + c_{ij} \cdot y_{ij} - d_{ij} \cdot z_{ij} .$$

The {\em worth of edge} $(i, j)$ is defined to be $w_{ij}$. We will define the {\em surplus} generated, when edge $(i, j)$ is matched, to be the right-hand-side of this inequality. The surplus adjusts the worth of the edge endogenously up or down depending on whether the lower or the upper bound constraint on $(i, j)$, respectively, is binding; if neither is binding the worth and surplus are the same. 

Thus, when the lower bound constraint on $(i, j)$ is binding, matching $i$ to $j$ generates more value than is physically possible in order to encourage more matchings of this type. Similarly,  when the upper bound constraint on $(i, j)$ is binding, matching $i$ to $j$ generates less value than is physically possible in order to prevent additional matchings of this type. 

As before, the {\em worth of the game} is the weight of a maximum weight HK matching in $G$ and is denoted by $W(U \cup V)$. However, the more important quantity is the {\em total surplus}, which is defined to be:

$$ \surplus(U \cup V) \ := \ W(U \cup V)  + \sum_{(i, j) \in E}  {(c_{ij} \cdot y_{ij} - d_{ij} \cdot z_{ij})} .$$

For a sub-coalition $(S_u \cup S_v)$, with $S_u \subseteq U, S_v \subseteq V$, let $G(S_u \cup S_v)$ denote the subgraph of $G$ induced on $S_u \cup S_v$ and let $E_S$ denote its edges. The worth of $(S_u \cup S_v)$ is the weight of a maximum weight HK matching in $G(S_u \cup S_v)$ and is denoted by $W(S_u \cup S_v)$. Next consider the dual LP (\ref{eq.Hoff-dual}) defined for the subgraph  $G(S_u \cup S_v)$ and consider an optimal solution for it. Let $y_{ij}^S$ and $z_{ij}^S$ be the dual variables for the lower and upper bound constraints on edge $(i, j) \in E_S$ in this optimal solution, respectively. Then, the {\em total surplus generated in $G(S_u \cup S_v)$} is defined to be:

$$  \surplus(S_u \cup S_v) \ := \ W(S_u \cup S_v)  + \sum_{(i, j) \in E}  {(c_{ij} \cdot y_{ij}^S - d_{ij} \cdot z_{ij}^S)} .$$

An {\em imputation} for the Hoffman-Kruskal game gives a way of dividing the total surplus among agents. Let $\alpha_i$ and $\beta_j$ denote the {\em payments} made to agents $i \in U$ and $j \in V$, respectively, under this imputation. Then, 

$$  \surplus(U \cup V) \ = \  {\sum_{i \in U}  {\alpha_i} + \sum_{j \in V} {\beta_j}}  . $$

An imputation $(\alpha, \beta)$ is in the {\em core of the Hoffman-Kruskal game} if for every sub-coalition $(S_u \cup S_v)$, 

$$  \surplus(S_u \cup S_v) \ \leq \  {\sum_{i \in S_u}  {\alpha_i} + \sum_{j \in S_v} {\beta_j}}  , $$
i.e., the total payment to the agents of the sub-coalition is at least as large as the surplus which the sub-coalition can generate.

\begin{theorem}
	\label{thm.core-Hoffman}
	For the Hoffman-Kruskal game, the dual partially characterizes its core in the following sense:
	Let $(u, v, y, z)$ be an optimal solution to the dual LP (\ref{eq.Hoff-dual}). Obtain $\surplus(U \cup V)$ using the worth of the game and these dual variables. Then $(\alpha, \beta)$, where $\alpha_i = b_i \cdot u_i$ and $\beta_j = b_j \cdot v_j$, is an imputation and is in the core of this game. 
\end{theorem}

\begin{proof}
The proof again hinges on the fact that the polytope defined by the constraints of the primal LP, (\ref{eq.Hoff-primal}) has integral vertices, which follows from Theorem \ref{thm.Hoffman}. 

Let $(u, v, y, z)$ be an optimal solution to the dual LP (\ref{eq.Hoff-dual}). By integrality and the LP-duality theorem, the worth of the game,  
$$W(U \cup V) = {\sum_{i \in U}  {b_i u_i} + \sum_{j \in V} {b_j v_j} + \sum_{(i, j) \in E}  {(d_{ij} z_{ij} - c_{ij} y_{ij})}} .$$
Using $\alpha_i = b_i \cdot u_i$ and $\beta_j = b_j \cdot v_j$, we get:  

$$ \surplus(U \cup V) \ = \ W(U \cup V)  + \sum_{(i, j) \in E}  {(c_{ij} \cdot y_{ij} - d_{ij} \cdot z_{ij})}  \ = \  {\sum_{i \in U}  {\alpha_i} + \sum_{j \in V} {\beta_j}}  , $$
and hence $(\alpha, \beta)$ is an imputation. 

Consider a sub-coalition $(S_u \cup S_v)$, with $S_u \subseteq U, S_v \subseteq V$. The restriction of the dual $(u, v, y, z)$ to the vertices and edges of $G(S_u \cup S_v)$ satisfies the constraints of the dual for this graph and is therefore a feasible dual for this subgraph. By integrality and the LP-duality theorem, we get that worth of $(S_u \cup S_v)$, namely $W(S_u \cup S_v)$, equals the objective function value of the optimal dual for this subgraph, and this is upper bounded by any feasible dual for this subgraph. Therefore, 

$$ W(S_u \cup S_v) \ \leq \ {\sum_{i \in S_u}  {b_i u_i} + \sum_{j \in S_v} {b_j v_j} + \sum_{(i, j) \in E_S}  {(d_{ij} z_{ij} - c_{ij} y_{ij})}} . $$
Therefore, 
$$  \surplus(S_u \cup S_v) \ \leq \  {\sum_{i \in S_u}  {\alpha_i} + \sum_{j \in S_v} {\beta_j}}, $$
thereby proving that $(\alpha, \beta)$ is a core imputation.
 \end{proof}

\begin{corollary}
\label{cor.b-gen}
	The core of the Hoffman-Kruskal game is always non-empty. 
\end{corollary}

By Example \ref{ex.b-arbitrary}, there are core imputations for the special case of the bipartite $b$-matching game which are not optimal dual solutions. Therefore the same holds for the Hoffman-Kruskal game. Remark \ref{rem.char} carries over to the Hoffman-Kruskal game. In particular, essential players are defined in Definition \ref{def.b-player} and core imputations corresponding to optimal dual solutions allocate all the surplus to essential players.

\subsection{Is the surplus of the Hoffman-Kruskal game uniquely defined?}
\label{sec.unique?}

In this section, we will answer the question stated in the title. Clearly, the worth of this game is unique. The answer\footnote{This non-uniqueness is similar to the non-uniqueness of other solution concepts in economics, e.g., Nash equilibrium and market equilibrium.} turns out to be ``No'', as shown in Example \ref{ex.not-unique}. Therefore, Theorem \ref{thm.core-Hoffman} holds only if the surplus and the payments to agents come from the {\em same} optimal dual solution. In the assignment game, different optimal duals make different allocations of profits to agents; however, the total profit allocated is the same and is the worth of the game. On the other hand, as shown in Example \ref{ex.not-unique}, the total surplus of different optimal duals can be different under the Hoffman-Kruskal game.

\begin{example}
	\label{ex.not-unique}
Consider the following instance of the Hoffman-Kruskal game, based on the graph of Figure \ref{fig.7}. Let the $b$ values of vertices be $(4, 2, 3)$ for $(u, v_1, v_2)$. For edges $(u, v_1)$ and $(u, v_2)$, let their weights be $1$ and $3$, and let their lower and upper bounds be $(1, 2)$ and $(0, 3)$, respectively. 

The maximum weight HK matching matches edges $(u, v_1)$ and $(u, v_2)$ 1 and 3 times, respectively, for a primal solution of value 10. Here are two optimal duals of value 10: The first assigns values of $(4, 0, 0)$ to vertices $(u, v_1, v_2)$, 2 for the lower bound for edge $(u, v_1)$ and the rest zero. The second assigns values of $(1, 0, 2)$ to vertices $(u, v_1, v_2)$ and zero to all edge dual variables. The surplus of the game under these two duals is 12 and 10, respectively. The payments to vertices $(u, v_1, v_2)$ under these two duals is $(12, 0, 0)$ and $(4, 0, 6)$, respectively.

\end{example}

\subsection{Are there other Matching-Based Games Having Complete Characterizations?}
\label{sec.complete}

Recall that for the uniform bipartite $b$-matching game, which is a sub-case of the Hoffman-Kruskal game, the optimal dual completely characterizes the core. This leads to the following interesting question:\\

{\em Are there other sub-cases of the Hoffman-Kruskal game for which the optimal dual completely characterizes the core?}

By Example \ref{ex.b-arbitrary}, if the answer were ``yes'', $b$ must be the constant function. Now there are two possibilities: there are edge upper bounds or there are edge lower bounds. To force a ``yes'' answer, let us pick these bounds to also be the constant function. Hence we are left with examining the following possibilities:
\begin{enumerate}
	\item  $b$ is the constant function and edge upper bounds are also constant; there are no edge lower bounds.
	\item  $b$ is the constant function and edge lower bounds are also constant; there are no edge upper bounds.
\end{enumerate}

Examples \ref{ex.HK-edge-upper} and \ref{ex.HK-edge-lower} give instances showing that for both these games, the optimal dual only partially characterizes the core.

\begin{example}
	\label{ex.HK-edge-upper}
Let the graph of Figure \ref{fig.7} define the following Hoffman-Kruskal game: the $b$ value of all vertices is $2$ and all edge upper bounds are 1. As before, let the edge weights be $1$ and $3$ for $(u, v_1)$ and $(u, v_2)$, respectively. 

The worth of the game is 4 and is obtained by matching both edges once each. The following are optimal dual solutions: 
\begin{enumerate}
	\item Assign zero to all vertices and 1 and 3 to edges $(u, v_1)$ and $(u, v_2)$, respectively. 
	\item Assign 1 to $u$, 2 to the upper bound dual of the edge $(u, v_2)$ and zero to the remaining vertices and edges. 
\end{enumerate}
Clearly, any convex combination of these two is also an optima dual solution. It is easy to verify that this list is exhaustive. 

The surplus of the game is zero under the first dual and 2 under the second. Corresponding to the second dual, the allocation of payments of $(1, 0, 1)$ for $(u, v_1, v_2)$ is a core imputation. The corresponding dual solution assigns $(1/2, 0, 1/2)$ to vertices $(u, v_1, v_2)$ and 2 to the upper bound dual of the edge $(u, v_2)$; this dual is infeasible. Therefore this core imputation does not correspond to an optimal dual solution.
\end{example}

\begin{example}
	\label{ex.HK-edge-lower}
Let the graph of Figure \ref{fig.7} define the following Hoffman-Kruskal game: the $b$ value of all vertices is $2$ and all edge lower bounds are 1; there are no edge upper bounds. As before, let the edge weights be $1$ and $3$ for $(u, v_1)$ and $(u, v_2)$, respectively. 

The worth of the game is 4 and is obtained by matching both edges once each. Consider the following optimal dual solution: 	 Assign 3 to $u$, 2 to the lower bound dual of the edge $(u, v_1)$ and zero to the remaining vertices and edges. Under this dual, the surplus of the game is 6 and the corresponding payments of $(6, 0, 0)$ for $(u, v_1, v_2)$ is a core imputation. 

Corresponding to this surplus, another payment in the core is $(3, 3, 0)$ for $(u, v_1, v_2)$. The corresponding dual solution assigns $(1.5, 0, 1.5)$ to vertices $(u, v_1, v_2)$ and 2 to the lower bound dual of the edge $(u, v_1)$; this dual is infeasible. Therefore this core imputation does not correspond to an optimal dual solution.
\end{example}

\section{The Core of Concurrent Games}  
\label{sec.general}

The {\em general graph matching game} consists of an undirected graph $G = (V, E)$ and an edge-weight function $w$, with the vertices $V$ being agents and edges $E$ being possible doubles teams in the tennis analogy of Section \ref{sec.intro}; $w_{i j}$ represents the profit generated by team $(i, j)$. The {\em worth} of a coalition $S \subseteq V$, denoted by $p(S)$, is defined to be the weight of a maximum weight matching in the graph $G(S)$, i.e., the restriction of $G$ to vertices in $S$. Definitions \ref{def.cooperative-game}, \ref{def.imputation}, \ref{def.core} carry over.

 Deng et al. \cite{Deng1999algorithms} showed that the core of this game is non-empty if and only if the weights of maximum weight integral and fractional matchings concur; if so, we will say that the game is {\em concurrent}. Below we give the LP-relaxation of the problem of finding a maximum weight fractional matching and thereby provide the underlying reason for the above-stated result of \cite{Deng1999algorithms} as well as their characterization of the core of concurrent games.  

We will work with the following LP (\ref{eq.core-primal}), whose optimal solutions are maximum weight fractional matchings in $G$. 

	\begin{maxi}
		{} {\sum_{(i, j) \in E}  {w_{ij} x_{ij}}}
			{\label{eq.core-primal}}
		{}
		\addConstraint{\sum_{(i, j) \in E} {x_{ij}}}{\leq 1 \quad}{\forall i \in V}
		\addConstraint{x_{ij}}{\geq 0}{\forall (i, j) \in E}
	\end{maxi}
	
Note that in case $G$ is bipartite, LP (\ref{eq.core-primal}) is equivalent to LP (\ref{eq.core-primal-bipartite}). Therefore, by Theorem \ref{thm.Hoffman}, it always has an integral optimal solution. On the other hand, if $G$ is non-bipartite, LP (\ref{eq.core-primal}) may have no integral optimal solutions, e.g., a triangle with unit weight edges. However, by Theorem \ref{thm.Balinski}, this LP always has a half-integral optimal solution.

\begin{theorem}
\label{thm.Balinski}  (Balinski \cite{Balinski1965integer}) 
	For a general graph, the vertices of the polytope defined by the constraints of LP (\ref{eq.core-primal}) are half-integral, such that the edges set to 1 form a matching and those set to half form disjoint odd-length cycles. 
\end{theorem}

Taking $v_i$ to be dual variables for the first constraint of (\ref{eq.core-primal}), we obtain  LP (\ref{eq.core-dual}). Any feasible solution to this LP is called a {\em cover} of $G$, since for each edge $(i, j)$, $v_i$ and $v_j$ cover edge $(i, j)$ in the sense that $v_i + v_j \geq w_{ij}$. An optimal solution to this LP is a {\em minimum cover}. We will say that $v_i$ is the {\em profit} of player $i$.

 	\begin{mini}
		{} {\sum_{i \in V}  {v_{i}}} 
			{\label{eq.core-dual}}
		{}
		\addConstraint{v_i + v_j}{ \geq w_{ij} \quad }{\forall (i, j) \in E}
		\addConstraint{v_{i}}{\geq 0}{\forall i \in V}
	\end{mini}
	
Let $Q_f$ ($Q_i$) be the weight of a maximum weight fractional (integral) matching in $G$. Now, $Q_f \geq Q_i$, since every integral matching is also a fractional matching. By the LP Duality Theorem, $Q_f$ equals the total value of a minimum cover. On the other hand, $Q_i$ is the worth of the game.

\cite{Deng1999algorithms} proved that the core of the general graph matching game is non-empty if and only if $Q_f = Q_i$. If so, by a proof that is similar to that of the Shapley-Shubik Theorem, it is easy to see each optimal solution to the dual LP, namely LP (\ref{eq.core-dual}),  gives a way for distributing the worth of the game among agents so that the condition of the core is satisfied, i.e., it is a core imputation. The converse is also true, i.e., every core imputation is an optimal solution to the dual LP. We summarize below.

\begin{theorem}
	\label{thm.Deng}
	(Deng et al. \cite{Deng1999algorithms}) 
	The core of the general graph matching game is non-empty if and only if $Q_f = Q_i$. If so, 	the dual LP (\ref{eq.core-dual}) completely characterizes the core of the game.  
\end{theorem}

\begin{example}
	\label{ex.5}
	Consider the graph given in Figure \ref{fig.5}. Assume that the weight of edge $(v_2, v_7)$ is 2 and the weights of the rest of the edges is 1. This graph has three maximum weight integral matching, of weight 4, and edge $(v_2, v_7)$ is in each of them; in addition, the three matchings contain the edge sets $\{(v_1, v_6), (v_4, v_5) \}, \{(v_1, v_6), (v_3, v_4) \}$ and $\{(v_3, v_4), (v_5, v_6) \}$. This graph also has three fractional matchings having a weight of 4, which are not integral: the seven-cycle $v_1, v_2, v_7, v_3, v_4, v_5, v_6$ taken half-integrally; the three-cycle $v_1, v_2, v_7$ taken half-integrally together with the edges $\{(v_3, v_4), (v_5, v_6)\}$; and the three-cycle $v_2, v_3, v_7$ taken half-integrally together with the edges $\{(v_4, v_5), (v_1, v_6)\}$. 
	
	By Theorem \ref{thm.Deng}, this game has a non-empty core. The unique core imputation assigns a profit of 1 to each of $v_2, v_4, v_6, v_7$ and zero to the rest. Observe that if the weight of edge $(v_2, v_7)$ is decreased a bit, then the core will become empty and if it is increased a bit, then the maximum weight fractional matchings will all be integral and the core will remain non-empty. 
\end{example}


\begin{figure}[h]
\begin{center}
\includegraphics[width=2.8in]{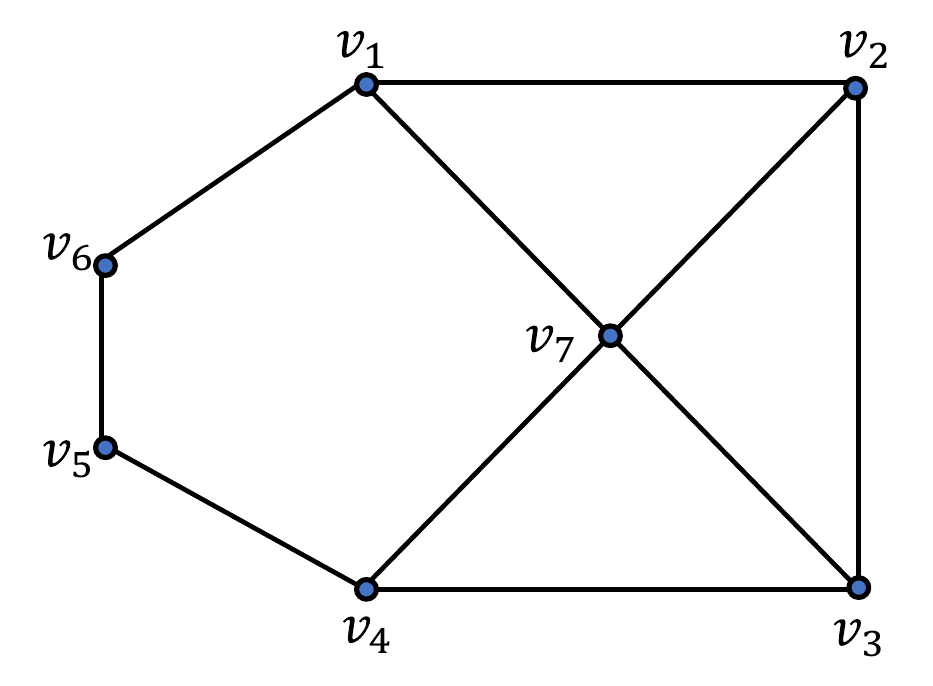}
\caption{The graph for Example \ref{ex.5}.}
\label{fig.5}
\end{center}
\end{figure}



For completeness, we describe a different LP for general graphs, namely LP (\ref{eq.core-primal-general}), given by Edmonds \cite{Edmonds.matching}, which always has integral optimal solutions; these are maximum weight matchings in the graph. This LP is an enhancement of LP (\ref{eq.core-primal}) via odd set constraints, as specified in (\ref{eq.core-primal-general}). The latter constraints are exponential in number, namely one for every odd subset $S$ of vertices. Clearly, the total number of integral matched edges in this set can be at most ${{(|S|-1)} \over 2}$. The constraint imposes this bound on any fractional matching as well, thereby ensuring that a fractional matching that picks each edge of an odd cycle half-integrally is disallowed and there is always an integral optimal matching. 

	\begin{maxi}
		{} {\sum_{(i, j) \in E}  {w_{ij} x_{ij}}}
			{\label{eq.core-primal-general}}
		{}
		\addConstraint{\sum_{(i, j) \in E} {x_{ij}}}{\leq 1 \quad}{\forall i \in V}
		\addConstraint{\sum_{(i, j) \in S} {x_{ij}}}{\leq {{(|S|-1)} \over 2} \quad}{\forall S \subseteq V, \ S \ \mbox{odd}}
		\addConstraint{x_{ij}}{\geq 0}{\forall (i, j) \in E}
	\end{maxi}

The dual of this LP has, in addition to variables corresponding to vertices, $v_i$, exponentially many more variables corresponding to odd sets, $z_S$, as given in (\ref{eq.core-dual-general}). As a result, the entire worth of the game does not reside on vertices only --- it also resides on odd sets.

 	\begin{mini}
		{} {\sum_{i \in V}  {v_{i}} + \sum_{S \subseteq V, \ \mbox{odd}} {z_S}} 
			{\label{eq.core-dual-general}}
		{}
		\addConstraint{v_i + v_j + \sum_{S \, \ni \, i, j}{z_S}}{ \geq w_{ij} \quad }{\forall (i, j) \in E}
		\addConstraint{v_{i}}{\geq 0}{\forall i \in V}
		\addConstraint{z_{S}}{\geq 0}{\forall S \subseteq V, \ S \ \mbox{odd}}
	\end{mini}

There is no natural way of dividing $z_S$ among the vertices in $S$ to restore core properties. The situation is more serious than that: it turns out that in general, the core of a non-bipartite game may be empty. 

This is easy to see for the graph $K_3$, i.e., a clique on three vertices, $i, j, k$, with a weight of 1 on each edge. Any maximum matching in $K_3$ has only one edge, and therefore the worth of this game is 1. Suppose there is an imputation $v$ which lies in the core. Consider all three two-agent coalitions. Then, we must have:
\[ v(i) + v(j) \geq 1, \ \ \ \  v(j) + v(k) \geq 1 \ \ \ \ \mbox{and} \ \ \ \ v(i) + v(k) \geq 1 .\]
This implies $v(i) + v(j) + v(k) \geq 3/2$ which exceeds the worth of the game, giving a contradiction.

One recourse to this eventuality was provided in \cite{Va.general}, which gives a ${2 \over 3}$-approximate core imputation for such games, i.e., the weight of a maximum weight matching in the graph is distributed among vertices in such a way that the total profit accrued by agents in a sub-coalition $S \subseteq V$ is at least ${2 \over 3}$ fraction of the profit which $S$ can generate by itself.

\subsection{Complementarity Applied to Concurrent Games}
\label{sec.Gen-Insights}

In this section, we will provide answers to the three issues raised in the Introduction, as they pertain to concurrent games. First we provide necessary definitions, which are adaptations of definitions from Section \ref{sec.Complementarity}. We will assume that $G = (V, E)$, $w: E \rightarrow \QQ_+$ is a concurrent game, 

\begin{definition}
	\label{def.Gen-player}
	A generic player in $V$ will be denoted by $q$. We will say that $q$ is:
	\begin{enumerate}
		\item {\em essential} if $q$ is matched in every maximum weight integral matching in $G$.
		\item {\em viable} if there is a maximum weight integral matching $M$ such that $q$ is matched in $M$ and another, $M'$ such that $q$ is not matched in $M'$. 	
		\item {\em subpar} if for every maximum weight integral matching $M$ in $G$, $q$ is not matched in $M$. 	
		\end{enumerate}
\end{definition}

\begin{definition}
\label{def.Gen-player-paid}
	Let $y$ be an imputation in the core. We will say that $q$ {\em gets paid in $y$} if $y_q > 0$ and {\em does not get paid} otherwise. Furthermore, $q$ is {\em paid sometimes} if there is at least one imputation in the core under which $q$ gets paid, and it is {\em never paid} if it is not paid under every imputation. 
\end{definition}

\begin{definition}
	\label{def.Gen-team}
	By a {\em team} we mean an edge in $G$; a generic one will be denoted as $e = (u, v)$. We will say that $e$ is:
	\begin{enumerate}
		\item {\em essential} if $e$ is matched in every maximum weight matching in $G$.
		\item {\em viable} if there is a maximum weight matching $M$ such that $e \in M$, and another, $M'$ such that $e \notin M'$. 
		\item {\em subpar} if for every maximum weight matching $M$ in $G$, $e \notin M$. 
	\end{enumerate}
	\end{definition}
	
\begin{definition}
\label{def.Gen-team-paid}
	 Let $y$ be an imputation in the core of a concurrent game. We will say that $e$ is {\em fairly paid in $y$} if $y_u + y_v = w_e$ and it is {\em overpaid} if $y_u + y_v > w_e$\footnote{Observe that by the first constraint of the dual LP (\ref{eq.core-dual-bipartite}), these are the only possibilities.}. Finally, we will say that $e$ is {\em always paid fairly} if it is fairly paid in every imputation in the core, and it is {\em sometimes overpaid} if there is a core imputation in which it is overpaid.
\end{definition}

\begin{theorem}
	\label{thm.Gen-insights}
	The following hold:
\begin{enumerate}
	\item For every player $q \in (U \cup V)$: 
		\[ q \ \mbox{is paid sometimes}  \ \Rightarrow \ q \ \mbox{is essential} \]
	\item 	 For every team $e \in E$: 
		\[ e \ \mbox{is viable or essential} \ \Rightarrow \ e \ \mbox{is always paid fairly} \]
\end{enumerate}  
\end{theorem}
	
Observe that the first statement of Theorem \ref{thm.Gen-insights} is equivalent to the forward direction of Theorem \ref{thm.vertices}, and the proof is also identical. The second statement of Theorem \ref{thm.Gen-insights} is equivalent to the reverse direction of Theorem \ref{thm.vertices} and again the proof is identical.

The proofs of the reverse direction of Theorem \ref{thm.vertices} and the forward direction of Theorem \ref{thm.edges} do not hold for general graphs. Counter-examples to these statements are given in Example \ref{ex.6} and Example \ref{ex.7}, respectively. In bipartite graphs, both these facts   follow from Theorem \ref{thm.int-assn-LP}. The latter fact does not hold for LP (\ref{eq.core-primal-general}); however, a weaker fact, given in Theorem \ref{thm.Balinski}, holds.

\begin{example}
	\label{ex.6}
Consider the game depicted in Figure \ref{fig.6}. Let the weights of edges $(v_1, v_2)$, $(v_2, v_3)$, $(v_3, v_1)$ and $(v_1, v_4)$ be $1.5, 1, 1.5$ and $1$, respectively. Then an optimal integral matching is $\{ (v_1, v_4), (v_2, v_3) \}$, having weight 2. The three-cycle $v_1, v_2, v_3$, taken to the extent of half, is a fractional matching of the same weight. Therefore, this graph has non-empty core. It has a unique core imputation which assigns profits of $1, {1 \over 2}, {1 \over 2}, 0$ to $v_1, v_2, v_3, v_4$, respectively. Since $v_4$ is essential but is not paid in the unique core imputation, this game provides a counter-example to the reverse direction of Theorem \ref{thm.vertices} in general graphs.  
\end{example}


\begin{figure}[h]
\begin{center}
\includegraphics[width=2.4in]{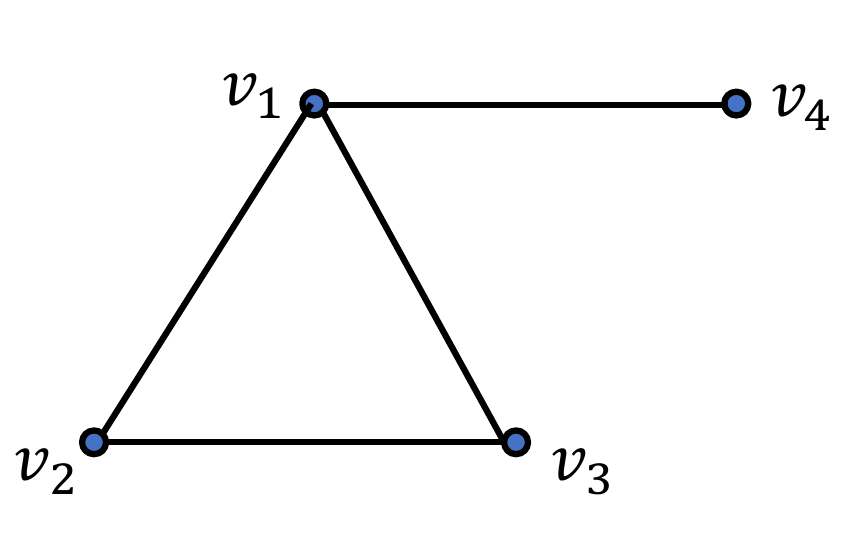}
\caption{The counter-example related to Theorem \ref{thm.Gen-insights}.}
\label{fig.6}
\end{center}
\end{figure}


Since the first statement of Theorem \ref{thm.Gen-insights} is weaker than Theorem \ref{thm.vertices}, the following corollary is weaker than Corollary \ref{cor.vertices}.  

\begin{corollary}
	\label{cor.Gen-vertices}
	In a concurrent game, the set of essential players is non-empty and in every core imputation, the entire worth of the game is distributed among essential players. 
\end{corollary} 

The contrapositive of the second statement of Theorem \ref{thm.Gen-insights} is the following:

 \begin{center}
{\em If a team is sometimes overpaid, then it is subpar.}
\end{center}

\begin{example}
	\label{ex.7}
	In the game defined in Example \ref{ex.5}, the teams $(v_4, v_7), (v_1, v_2), (v_2, v_3)$ and $(v_1, v_7)$ are all subpar, since they are not matched in any maximum weight matching. Of these, only the first team is overpaid in the unique imputation in the core, the rest are fairly paid. On the other hand,	the vertices which get positive profit are precisely the essential vertices. 
\end{example}
 
Corollary \ref{cor.Gen-degen} is identical to Corollary \ref{cor.degen}.

\begin{corollary}
	\label{cor.Gen-degen}
	In the presence of degeneracy, imputations in the core of a concurrent game treat:
	\begin{itemize}
			\item  viable players in the same way as subpar players, namely they are never paid.  
		\item viable teams in the same way as essential teams, namely they are always fairly paid. 
	\end{itemize}
\end{corollary}

\section{The Core of the Max-Flow Game} 
\label{sec.flow}

In Section \ref{sec.flow-LP}, we define the maximum flow problem and give its LP-formulation. The constraint matrix of this LP is TUM. However, unlike matching-based games, in which the constraint matrix had only $0/1$ entries, in this game, the constraint matrix has $-1$ entries as well. In Section \ref{sec.game-flow}, we define the max-flow game and give a partial characterization of its core.

\subsection{The Maximum Flow Problem and its LP-Formulation}
\label{sec.flow-LP}

Let $G=(V, E)$ be a directed graph with two distinguished vertices, a {\em source} $s$ and a {\em sink} $t$, and positive capacities, $c:E\rightarrow \Rplus$. A {\em flow} is a function $f: E \rightarrow \Rplus$ satisfying the following two constraints: 
\begin{enumerate}
\item
{\em Capacity constraint:} For each edge $e \in E$, the flow sent through $e$ is bounded by its capacity, $c_e$.
\item
{\em Flow conservation:} At each vertex $v$, other than $s$ and $t$, the total flow into $v$ should
equal the total flow out of $v$.
\end{enumerate}

The {\em value of flow $f$} is defined to be the total flow out of $s$ or the total flow into $t$; the two must be equal because of flow conservation at the rest of the vertices. The {\em maximum flow problem} is to find a flow whose value is maximum subject to these constraints. 

Partition the vertices into two sets $X$ and $\overline{X}$ so that $s \in X$ and $t \in \overline{X}$. Then the set of edges going from $X$ to $\overline{X}$ is called an {\em $s$--$t$ cut} and is denoted by $(X,\overline{X})$; its {\em capacity} is defined to be the sum of capacities of all these edges. Clearly, the capacity of any $s$--$t$ cut is an upper bound on any feasible flow. Therefore if the capacity of an $s$--$t$ cut, say $(X,\overline{X})$, equals the value of a feasible flow, $f$, then $(X,\overline{X})$ must be a {\em minimum $s$--$t$ cut in $G$} and $f$ must be a maximum flow in $G$. The celebrated {\em Max-Flow Min-Cut Theorem} proves that it is always possible to find a flow and an $s$--$t$ cut so that equality holds; we provide a proof sketch below. 

Let us formulate the maximum flow problem as a linear program. In order to obtain a simple formulation, we will first introduce a new edge of unbounded capacity from $t$ to $s$ and introduce flow conservation at $s$ and $t$ as well, thus converting the flow into a circulation. To obtain a maximum flow in the original graph, we seek a circulation which maximizes the flow $f_{ts}$ on the new  edge. Let primal variable $f_{ij}$ denote the amount of flow sent through edge $(i,j)\in E$. The primal linear program is the following:

	\begin{maxi}
		{} {f_{ts}}
			{\label{eq.flow-primal}}
		{}
		\addConstraint{\sum_{(j,i)\in E}{f_{ji}} ~-
            \sum_{(i,j)\in E}{f_{ij}}}{\leq 0 \quad}{\forall i \in V}
		\addConstraint{f_{ij}}{\leq c_{ij} }{\forall (i,j) \in E}
		\addConstraint{f_{ij}}{\geq 0}{\forall (i, j) \in E}
	\end{maxi}

The first set of inequalities say that for each node $i$, the total flow into $i$ is at most the total flow out of $i$. In order to obtain a maximization LP in standard form, we have not made the in-flow {\em equal to} the out-flow; however, as shown next, equality is guaranteed. 

\begin{lemma}
	\label{lem.equality}
If the inequality in first constraint of LP (\ref{eq.flow-primal}) holds at each vertex $v \in V$, then in fact it is satisfied with equality at each vertex. 
\end{lemma}

\begin{proof}
Suppose not and assume that the inequality is strict for at least one vertex. Consider the sum of all these $|V|$ constraints. On the LHS, all terms cancel out and we get 0; the RHS is also 0. Therefore, we get $0 < 0$, which is a contradiction. Hence the lemma holds. 
\end{proof}

Next, let us introduce dual variables $\pi_i$ and $\delta_{ij}$ for the first and second constraints, respectively, of the primal LP to obtain the dual program.

 	\begin{mini}
		{} {\sum_{(i, j) \in E}  {c_{ij} \delta_{ij}}}  
			{\label{eq.flow-dual}}
		{}
		\addConstraint{\delta_{ij} - \pi_i + \pi_j}{ \geq 0 \quad }{\forall (i, j) \in E}
		\addConstraint{\pi_s - \pi_t}{\geq 1}{}
		\addConstraint{\delta_{ij}}{\geq 0}{\forall (i, j) \in E}
		\addConstraint{\pi_i}{\geq 0}{\forall i \in V}
	\end{mini}

It will be convenient to view variables $\delta_{ij}$ and $\pi_i$ as distance labels on edges and potentials on vertices, respectively. In order to gain intuition on their use, let us first assume that for an instance $(G, c)$, the dual LP (\ref{eq.flow-dual}) has an {\em integral} optimal solution. It is easy to see that all variables in this solution must be set to 0 or 1, since no variable needs to exceed 1. If so, $\pi_s = 1, \pi_t = 0$ and the potentials of the rest of the vertices are either 0 or 1. 

Let $X$ be the set of vertices having potential 1 and $\overline{X}$ be the ones having potential 0. Let $(i, j)$ be an edge in the $s$--$t$ cut $(X,\overline{X})$. By the first constraint, $\delta_{ij} \geq \pi_i + \pi_j = 1$. Since $\delta_{ij}$ is $0/1$, it must be 1. Since $\delta_{ij}> 0$, by Lemma \ref{lem.d>0}, $f_{ij} = c_{ij}$, where $f$ is a maximum flow. Therefore the objective function value of the dual is the capacity of cut $(X, \overline{X})$ and it equals the value of flow $f$. Therefore, $(X, \overline{X})$ is a minimum $s$--$t$ cut in $G$. Any path from $s$ to $t$ must encounter an edge in the $s$--$t$ cut $(X,\overline{X})$ and therefore under this dual solution, the distance labels on any $s$--$t$ path must add up to at least 1. 

The next lemma is analogous to Theorem \ref{thm.vertices}; it will be useful in this section as well as in Section \ref{sec.flow-fair}.

\begin{definition}
	\label{def.flow-essential}
We will say that an edge $(i, j) \in E$ is {\em essential} if it is saturated in every maximum flow.
\end{definition}

\begin{lemma}
	\label{lem.d>0}
For every player $(i, j) \in E$:
		\[ \delta_{ij} > 0 \ \mbox{in some optimal dual solution}  \ \iff \ (i, j) \ \mbox{is essential} . \]  
\end{lemma}

\begin{proof}
 Let $f$ and $(\delta, \pi)$ be optimal solutions to LP (\ref{eq.flow-primal}) and LP (\ref{eq.flow-dual}), respectively. By the Complementary Slackness Theorem, for each $(i, j) \in E: \ \delta_{ij} \cdot (f_{ij} - c_{ij}) = 0$. 

$(\Rightarrow)$  Suppose $\delta_{ij} > 0$ in some optimal dual. Now, by the Complementary Slackness Theorem, for any optimal primal solution, $f$, $f_{ij} = c_{ij}$, i.e., $(i, j)$ is saturated in every max-flow and is therefore essential. This proves the forward direction.

$(\Leftarrow)$ Strict complementarity implies that corresponding to each player $(i, j)$, there is a pair of optimal primal and dual solutions, say $f$ and $(\delta, \pi)$, such that either $\delta_{ij} = 0$ or $f_{ij} = c_{ij}$ but not both. Assume that $(i, j)$ is essential, i.e., $f_{ij} = c_{ij}$ is every optimal primal solution. Then by strict complementarity there is an optimal dual such that $\delta_{ij} > 0$. This proves the reverse direction. 
\end{proof}

Next, let us understand how to interpret a {\em fractional} optimal solution to the dual LP (\ref{eq.flow-dual}); this will be critically used in Section \ref{sec.game-flow}. Let $(\delta, \pi)$ be a feasible dual solution and consider an $s$--$t$ path $(s = v_1, v_2, \ldots, v_l = t)$. By the first constraint in the dual, for each edge $(i, i+1)$ on this path, $\delta_{i, i+1} \geq  \pi_i - \pi_{i+1}$. Therefore, the sum of the distance labels on this path is 

$$ \sum_{i=1}^{l-1} {\delta_{i, i+1}} \ \geq \  \sum_{i=1}^{l-1} {(\pi_i - \pi_{i+1})} = \pi_s - \pi_t \geq 1  . $$

Effectively, the distance labels hold all the information contained in a feasible dual solution, since we can set $\pi_s = 1, \pi_t = 0$ and compute the potentials of the rest of the vertices by using the first constraint of the dual LP. Furthermore, any distance labels such that on each $s$--$t$ path, the sum of the distance labels is at least 1, constitute a feasible dual solution. 

We will define the distance labels of an arbitrary optimal dual solution to be a {\em fractional $s$--$t$ cut}; its {\em capacity} is defined to be the objective function value achieved by the dual, i.e., $ {\sum_{(i, j) \in E}  {c_{ij} \delta_{ij}}}$. By the LP-Duality Theorem, this capacity is equal to the max-flow value.

In principle, the minimum fractional $s$--$t$ cut could have lower capacity than the minimum integral cut. It is easy to prove via induction, see \cite{LP.book}, that the constraint matrix of LP (\ref{eq.flow-dual}) is TUM and since the constants in the constraints of the dual are integral, the underlying polytope of this LP is integral and hence LP (\ref{eq.flow-dual}) always has an integral optimal solution; the latter is a minimum $s$--$t$ cut. This proves the Max-Flow Min-Cut Theorem.

\subsection{The Max-Flow Game and a Characterization of its Core}
\label{sec.game-flow}

The {\em max-flow game} is defined over an instance of the maximum flow problem; assume the latter  consist of directed graph $G=(V, E)$, with two distinguished vertices, a {\em source} $s$ and a {\em sink} $t$, and non-negative capacities, $c:E\rightarrow \Rplus$. The {\em agents} of the game are the edges. Flow sent from $s$ to $t$ generates profit which needs to be divided up in a ``fair'' manner among the agents. 

The grand coalition is $E$ and the {\em worth of the game}, $\worth(E)$, is defined to be the value of the maximum flow in $G$. For any sub-coalition $S \subseteq E$, let $G(S)$ denote the graph $G = (V, S)$, i.e., the vertex set is $V$ the graph has the edges of $S$ only. The capacities on these edges are given by the given by the restriction of function $c$ to $S$, i.e., $c:S \rightarrow \Rplus$. The {\em worth of $S$}, $\worth(S)$, is defined to be the value of a maximum flow in $G(S)$. The definitions of imputation and core given in Definitions \ref{def.imputation} and \ref{def.core} carry over.

\begin{theorem}
	\label{thm.flow}
	For the max-flow game, the dual partially characterizes its core.
\end{theorem}

\begin{proof}
Let $(\delta, \pi)$ be an optimal dual solution. By the LP-duality theorem, the worth of the game equals the objective function value of an optimal dual, i.e.,  
$$\worth(E) =  {\sum_{(i, j) \in E}  {c_{ij} \cdot \delta_{ij}}} .$$  

Define the {\em profit} of edge $(i, j)$ to be 
$$ p(i, j) := {c_{ij} \cdot \delta_{ij}} .$$
Clearly, $p$ is an imputation and by Lemma \ref{lem.d>0}, only fully saturated edges can have positive profit. For any sub-coalition $S \subseteq V$, define 
$$ \profit(S) := \sum_{e \in S} {p(e)} .$$
To prove that $p$ is in the core of the max-flow game, we need to show that for any sub-coalition $S \subseteq V$, $\worth(S) \leq \profit(S)$. 
	 
Consider a sub-coalition $S \subseteq V$. Let $f$ be a max-flow in the graph $G(S)$; its value is $\worth(S)$. Using standard methods, $f$ can be decomposed into at most $|S|$ flow paths. Let $\ccP$ denote the set of all such paths and for path $p \in \ccP$, let $f_p$ denote the flow sent on $p$. 

We will use the following two facts. First, as stated above, the sum of distance labels on $p$, $\sum_{e \in p} {d_e} \geq 1$. Second, by Lemma \ref{lem.d>0}, if $d_e > 0$, $e$ is fully saturated, and therefore for such an edge, $\sum_{p \ni e} {f_p} \ = \ c_e$. Now,
$$ \worth(S) \ = \ \sum_{p \in \ccP} {f_p} \ \leq \ \sum_{p \in \ccP} {f_p \cdot \left(\sum_{e \in p} {d_e}\right)} \ = \ \sum_{e \in S} {d_e \cdot \left(\sum_{p \ni e} {f_p} \right)} \ = \ \sum_{e \in S} {d_e \cdot c_e} \ = \profit(S) , $$ 
where the inequality follows from the first fact and the second-last equality follows from the second fact. 

Example \ref{ex.flow-no-core} gives a max-flow game having a core imputation which does not correspond to an optimal dual solution. Therefore, the dual partially characterizes the core of the max-flow game.
\end{proof}


\begin{figure}[h]
\begin{center}
\includegraphics[width=4in]{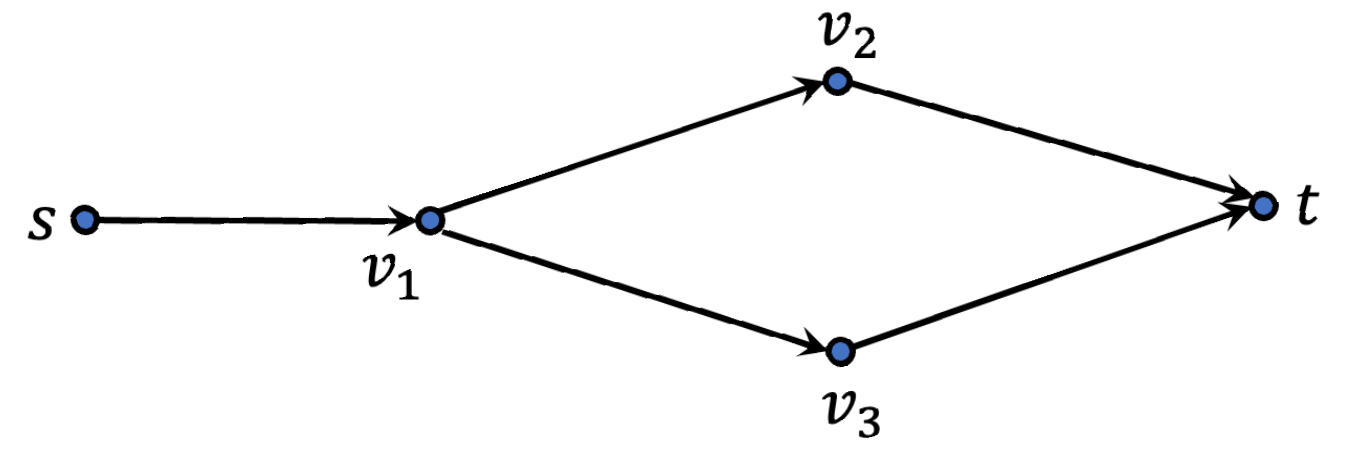}
\caption{The graph for Example \ref{ex.flow-no-core}.}
\label{fig.flow}
\end{center}
\end{figure}


\begin{corollary}
\label{cor.flow}
	The core of the max-flow game is always non-empty. 
\end{corollary}

\begin{lemma}
	\label{lem.flow-essential}
For the max-flow game, core imputations corresponding to optimal dual solutions allocate all the profits to essential players only.
\end{lemma}

\begin{proof}
By Lemma \ref{lem.d>0}, in an optimal dual solution, if $d_{ij} > 0$, then edge $(i, j)$ is essential. The lemma follows. 
\end{proof}

\begin{example}
	\label{ex.flow-no-core}
Consider the graph given in Figure \ref{fig.flow}. Assume that the capacities of edges $(v_1, v_2)$ and $(v_1, v_3)$ are 1 each. An easy way of showing that there is a core imputation which does not correspond to an optimal dual solution is the following: Assume that the rest of the edges have very high capacities, e.g., 10 each. Assign a profit of 2 to edge $(s, v_1)$ and zero to the rest; observe that the worth of the game is 2. Clearly, the dual will not assign edge $(s, v_1)$ a positive distance label, since a max-flow does not saturate it. 

To obtain a more ``principled'' example, we will insist that positive profits be assigned to fully saturated edges only. This is in keeping with the principle in general equilibrium theory that a good can have a positive price only if it is fully sold. 

Let us change the instance as follows: drop the capacity of $(s, v_1)$ to 2, keeping the rest as above. The worth of this game is also 2. Assign profits of $2 \epsilon, 1$ and $1 - 2 \epsilon$ to edges $(s, v_1), (v_1, v_2), (v_1, v_3)$, respectively, for a small number $\epsilon > 0$. It is easy to see that this imputation is in the core and it assigns profits to fully saturated edges only. If it were to correspond to a valid dual, the distance labels of these three edges would be $\epsilon, 2, 1 - 2 \epsilon$, respectively. However, that is not a feasible dual, since the  distance labels on the flow path $s, v_1, v_3, t$ add up to $1 -  \epsilon < 1$.  
\end{example}

\subsection{The Integral Maximum Flow Game and a Characterization of its Core}
\label{sec.game-int-flow}

The {\em integral maximum flow problem} is defined as follows: Given an instance of the maximum flow problem with {\em integral} capacities, $c:E\rightarrow \ZZ_+$, the problem is to find a maximum flow function that is integral, i.e., find $f: E \rightarrow \ZZ_+$ such that $f$ is a maximum flow. Such a flow is guaranteed to exist by the following: Since the constraint matrix of LP (\ref{eq.flow-primal}) is TUM and the capacity function $c$ is integral, by Theorem \ref{thm.Hoffman}, due to Hoffman and Kruskal, the underlying polytope of this LP is integral.  Hence there is a max-flow that is integral\footnote{This fact carries over to rational capacities, since they can be scaled to integral capacities by multiplying by the lowest common multiple of the denominators of edge capacities.}.

The {\em integral max-flow game} is defined over an instance of the integral maximum flow problem. The rest of the definitions are identical to those of the max-flow game given in Section \ref{sec.game-flow}. The proof of Theorem \ref{thm.int-flow} is analogous to the first half of the proof of Theorem \ref{thm.b-uniform}.

\begin{theorem}
	\label{thm.int-flow}
	For the integral max-flow game, the dual partially characterizes its core.
\end{theorem}

\begin{proof}
The proof hinges on the fact that the polytope defined by the constraints of the primal LP, (\ref{eq.flow-primal}) has integral vertices, i.e., they are integral flows. Let $(\delta, \pi)$ be an optimal dual solution. By integrality and the LP-duality theorem, the worth of the game, 
$$\worth(E) =  {\sum_{(i, j) \in E}  {c_{ij} \cdot \delta_{ij}}} .$$  

Define the {\em profit} of edge $(i, j)$ to be 
$$ p(i, j) := {c_{ij} \cdot \delta_{ij}} .$$
Clearly, $p$ is an imputation. For any sub-coalition $S \subseteq V$, define 
$$ \profit(S) := \sum_{e \in S} {p(e)} .$$
To prove that $p$ is in the core of the integral max-flow game, we need to show that for any sub-coalition $S \subseteq V$, $\worth(S) \leq \profit(S)$. 
	 
Consider a sub-coalition $S \subseteq V$. Then $\worth(S)$ is the value of a maximum integral flow in the graph $G(S)$. By integrality and LP-duality, $\worth(S)$  equals the objective function value of the optimal dual for graph $G(S)$. Since the restriction of $(\delta, \pi)$ to $G(S)$ is a feasible dual for $G(S)$, we get that 
$$  {\sum_{(i, j) \in S}  {c_{ij} \cdot \delta_{ij}}} = \profit(S) \geq \worth(S)  .$$  
Therefore $(\delta, \pi)$ is a core imputation.

Since Example \ref{ex.flow-no-core} is an instance of the integral max-flow game as well, therefore, the dual partially characterizes the core of the integral max-flow game.
\end{proof}

\begin{corollary}
\label{cor.int-flow}
	The core of the integral max-flow game is always non-empty. 
\end{corollary}

\section{Min-Max Fair, Max-Min Fair and Equitable Core Imputations}
\label{sec.fair}

We will denote an instance of the game by $I$ and the set of agents of instance $I$ by $A$. An imputation will be denoted by $p$; thus $p: A \rightarrow \Rplus$. Let $C(I)$ denote the set of core imputations which correspond to optimal dual solutions; if the dual completely characterizes the core, then $C(I)$ is the same as the core of the game and if it partially characterizes the core, it is a subset of the core.

\begin{definition}
	\label{def.min-max}
Let $I$ be an instance of the game. An imputation $p$ for $I$ is said to be a {\em min-max fair core imputation in $C(I)$} if $p \in C(I)$ and it satisfies:
$$ p \in \arg \min_{q \in C(I)} \left\{ \max_{a \in A} \{q(a)\} \right\} .$$
\end{definition}

Before defining max-min fair core imputations, we need to clarify the role that {\em essential agents} must play in it. Fact \ref{fact.essential} is a summary of results established about essential players in this paper. 

\begin{fact}
	\label{fact.essential}
For all games studied in this paper, for an instance $I$, each core imputation in $C(I)$ allocates all profits to essential agents only.
\end{fact}

By Fact \ref{fact.essential}, in all core imputations in $C(I)$, non-essential agents get zero profits. Therefore if the instance has non-essential agents, then under a straightforward definition of max-min fair core imputation, each imputation in $C(I)$ is trivially a max-min fair core imputation, with the max-min value being zero. Below is a more useful definition of max-min fair core imputation.

\begin{definition}
	\label{def.max-min}
Let $A' \subseteq A$ be the set of essential players of instance $I$. Imputation $p$ is said to be a {\em max-min fair core imputation in $C(I)$} if $p \in C(I)$ and it satisfies:
$$ p \in \arg \max_{q \in C(I)} \left\{ \min_{a \in A'} \{q(a)\} \right\} .$$
\end{definition}

\begin{definition}
	\label{def.equitable}
Let $A' \subseteq A$ be the set of essential players of instance $I$. Define the {\em spread} of imputation $p$ to be the difference between the largest and smallest profit to essential players, i.e., 

$$  \spread(p)  :=  (p_{\max} - p_{\min}), $$ 
$$ \mbox{where} \ \ p_{max} = \max_{a \in A'} {p(a)} \ \ \mbox{and} \ \ p_{min} = \min_{a \in A'}  {p(a)} .$$

Imputation $p$ is said to be an {\em equitable core imputation in $C(I)$} if $p \in C(I)$ and it satisfies:
$$ p \in \arg \min_{q \in C(I)} \left\{  \spread(q) \right\} .$$
\end{definition}

{\bf Procedure for Finding the Set of Essential Players:} 
For finding max-min fair and equitable core imputations, we will first need to compute the set of essential players of the game. We show how this can be done efficiently for all games studied in this paper. For the given assignment game, let $W_{\max}$ be its worth. Iteratively, for each player $u$, consider the game with player $u$ removed and find its worth. If it is the same as $W_{max}$ then $u$ is not essential and otherwise it is. For the rest of the bipartite generalizations of the assignment game, Definition \ref{def.b-player} applies. The only change needed to each iteration of the procedure is to drop $u$'s $b$ value by 1 instead of removing it. The procedure for concurrent and the max-flow games is identical to that for the assignment game.

\subsection{Fair Core Imputations for Matching-Based Games}
\label{sec.fair-matching}

As shown in Theorem \ref{thm.vertices}, under each core imputation, the entire worth of an assignment game is allocated to the essential players; for the latter notion, see Definition \ref{def.player}. Let $U' \subseteq U$ and $V' \subseteq V$ be the set of essential players and let $W_{\max}$ be the worth of the given instance of the assignment game. In Theorem \ref{thm.flow-assignment} we prove that LPs (\ref{eq.min-assignment}), (\ref{eq.max-assignment}) and (\ref{eq.equitable-assignment}  find min-max fair, max-min fair and equitable core imputations, respectively, for the assignment  game.

 	\begin{mini}
 			{} {{\alpha}}
			{\label{eq.min-assignment}}
		{}
		\addConstraint{\sum_{i \in U}  {u_{i}} + \sum_{j \in V} {v_j}}{= W_{\max}}{}
		\addConstraint{ u_i + v_j}{ \geq w_{ij} \quad }{\forall (i, j) \in E}
		\addConstraint{u_{i}}{\leq \alpha}{\forall i \in U'}
		\addConstraint{v_{j}}{\leq \alpha}{\forall j \in V'}
		\addConstraint{u_{i}}{\geq 0}{\forall i \in U'}
		\addConstraint{v_{j}}{\geq 0}{\forall j \in V'}
		\addConstraint{u_{i}}{= 0}{\forall i \in (U-U')}
		\addConstraint{v_{j}}{= 0}{\forall j \in (V-V')}
	\end{mini}

\begin{theorem}
	\label{thm.flow-assignment}
For the assignment game:
\begin{enumerate}
\item  An optimal solution to LP (\ref{eq.min-assignment}) is a min-max fair core imputation in $C(I)$.
\item  An optimal solution to LP (\ref{eq.max-assignment}) is a max-min fair core imputation in $C(I)$. 
\item  An optimal solution to LP (\ref{eq.equitable-assignment}) is an equitable core imputation in $C(I)$. 
\end{enumerate}
\end{theorem}

\begin{proof}
{\bf 1).} Consider an optimal solution $(u, v)$ to LP (\ref{eq.min-assignment}). We will first show that it is an optimal solution to the dual LP (\ref{eq.core-dual-bipartite}) as well. Constraints 7 and 8, for non-essential players, are justified by Theorem \ref{thm.vertices}. Furthermore, by constraints 2, 5 and 6, $(u, v)$ is a feasible solution to LP (\ref{eq.core-dual-bipartite}).  Finally, by the first constraint, the solution is an optimal solution to LP (\ref{eq.core-dual-bipartite}). 

To complete the proof, observe that constraints 3 and 4, and the minimization of $\alpha$ in the objective, ensure that the maximum profit of essential players is  minimized. 

{\bf 2).} First notice that constraints 3 and 4 in LP (\ref{eq.max-assignment}) make the non-negativity  constraints on essential players redundant, justifying their removal. Therefore, the only difference between the two LPs is that the first minimizes $\alpha$ and its third and fourth constraints ensure that the maximum profit of an essential players is minimized; and the second maximizes $\beta$ and its third and fourth constraints ensure that the minimum profit of an essential player is maximized. This completes the proof. 

{\bf 3).} LP (\ref{eq.equitable-assignment}) combines the constraints of the previous two LPs and its objective function minimizes $(\alpha - \beta)$, i.e., the spread of the imputation. Therefore, its optimal solution is an equitable core imputation in $C(I)$. 

\end{proof}

 	\begin{maxi}
 			{} {{\beta}}
			{\label{eq.max-assignment}}
		{}
		\addConstraint{\sum_{i \in U}  {u_{i}} + \sum_{j \in V} {v_j}}{= W_{\max}}{}
		\addConstraint{ u_i + v_j}{ \geq w_{ij} \quad }{\forall (i, j) \in E}
		\addConstraint{u_{i}}{\geq \beta}{\forall i \in U'}
		\addConstraint{v_{j}}{\geq \beta}{\forall j \in V'}
		\addConstraint{u_{i}}{= 0}{\forall i \in (U-U')}
		\addConstraint{v_{j}}{= 0}{\forall j \in (V-V')}
	\end{maxi}

 	\begin{mini}
 			{} {(\alpha - \beta)}
			{\label{eq.equitable-assignment}}
		{}
		\addConstraint{\sum_{i \in U}  {u_{i}} + \sum_{j \in V} {v_j}}{= W_{\max}}{}
		\addConstraint{ u_i + v_j}{ \geq w_{ij} \quad }{\forall (i, j) \in E}
		\addConstraint{u_{i}}{\leq \alpha}{\forall i \in U'}
		\addConstraint{v_{j}}{\leq \alpha}{\forall j \in V'}
		\addConstraint{u_{i}}{\geq \beta}{\forall i \in U'}
		\addConstraint{v_{j}}{\geq \beta}{\forall j \in V'}
		\addConstraint{u_{i}}{\geq 0}{\forall i \in U'}
		\addConstraint{v_{j}}{\geq 0}{\forall j \in V'}
		\addConstraint{u_{i}}{= 0}{\forall i \in (U-U')}
		\addConstraint{v_{j}}{= 0}{\forall j \in (V-V')}
	\end{mini}

\begin{example}
	\label{ex.Cross}
In the assignment game depicted in Figure \ref{fig.Cross}, weights are given on edges. The min-max fair and max-min fair  core imputations for vertices $(u_1, u_2, v_1, v_2)$ are $(50, 30, 50, 10)$ and $(40, 20, 60, 20)$, respectively. It is easy to see that both these imputations are equitable, each  with a spread of 40.  
\end{example}


\begin{figure}[h]
\begin{center}
\includegraphics[width=2.4in]{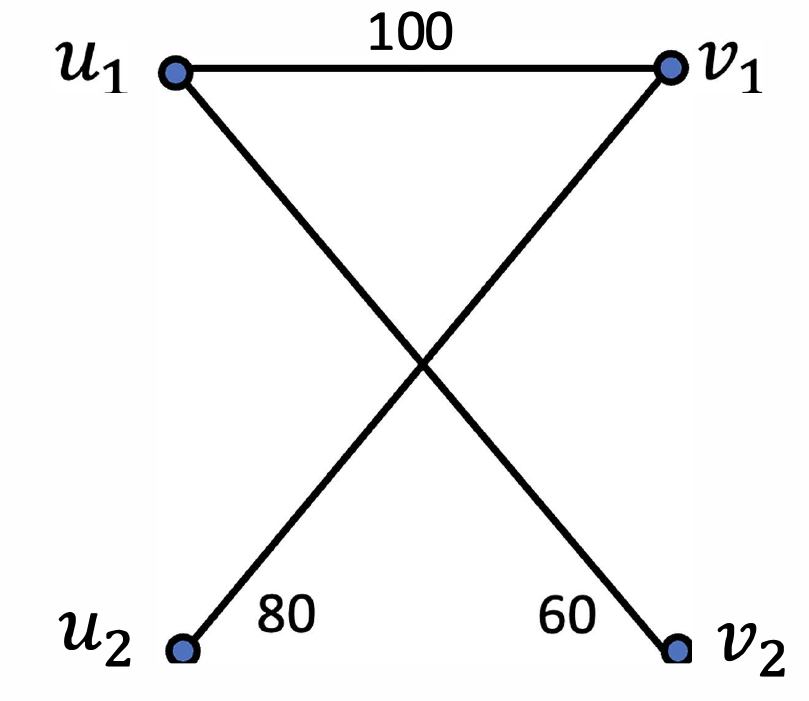}
\caption{The graph for Example \ref{ex.Cross}}.  
\label{fig.Cross}
\end{center}
\end{figure}


For a concurrent game, the dual LP (\ref{eq.core-dual}) gives core imputations and since it has the same for as that of an assignment game, and only essential players get profits by Theorem \ref{thm.Gen-insights}, min-max fair, max-min fair and equitable  core imputations follow easily. 

Next, we turn to the uniform bipartite $b$-matching game and the bipartite $b$-matching game, defined in Section \ref{sec.b-matching-game}. The definition of essential players changes to Definition \ref{def.b-player} and this changes the definition of $U'$ and $V'$. $W_{\max}$ is now the worth of the given instance of the $b$-matching game. With these changes to LPs (\ref{eq.min-assignment}) and (\ref{eq.max-assignment}), we again obtain min-max fair, max-min fair and equitable  core imputations.

Finally, we turn to the Hoffman-Kruskal game defined in Section \ref{sec.Hoffman}. A crucial difference from the previous games is that for this game, an imputation distributes surplus instead of worth. The definition of essential players is the same as in Definition \ref{def.b-player} and again, only essential players are allocated all surplus. 
 
LP (\ref{eq.min-Hoffman}) computes a min-max fair core imputation. For obtaining this LP, first find an optimal solution to the dual LP (\ref{eq.Hoff-dual}), say $(u, v, y, z)$. This solution is used for computing $\surplus(U \cup V)$, which is used in the first constraint of the new LP. The second constraint uses the settings of $y$ and $z$ from this dual solution. LPs for computing max-min  fair and equitable core imputations follow in a similar manner.

 	\begin{mini}
 			{} {{\alpha}}
			{\label{eq.min-Hoffman}}
		{}
		\addConstraint{\sum_{i \in U}  {u_{i}} + \sum_{j \in V} {v_j}}{= \surplus(U \cup V)}{}
		\addConstraint{ u_i + v_j}{ \geq w_{ij} + c_{ij} \cdot y_{ij} - d_{ij} \cdot z_{ij} \quad }{\forall (i, j) \in E}
		\addConstraint{u_{i}}{\leq \alpha}{\forall i \in U'}
		\addConstraint{v_{j}}{\leq \alpha}{\forall j \in V'}
		\addConstraint{u_{i}}{\geq 0}{\forall i \in U'}
		\addConstraint{v_{j}}{\geq 0}{\forall j \in V'}
		\addConstraint{u_{i}}{= 0}{\forall i \in (U-U')}
		\addConstraint{v_{j}}{= 0}{\forall j \in (V-V')}
	\end{mini}

\subsection{Fair Core Imputations for the Max-Flow Game}
\label{sec.flow-fair}

The max-flow game was defined by Kalai and Zemel \cite{Kalai1982totally} and they observed that a minimum cut defines a core imputation -- by paying each edge to the extent of its capacity. Clearly, an imputation corresponding to a fractional optimal dual solution is ``more fair'', since it spreads the profit over more players, giving each less than its maximum possible profit; for an extreme examples, see Example \ref{ex.fair-flow}.

\begin{example}
	\label{ex.fair-flow}
Consider a max-flow game over a graph which is a path of length $n$ from $s$ to $t$ of unit capacity edges. The worth of this game is 1 and the Kalai-Zemel imputation will give the entire profit to one of the edges; each one is a minimum cut. In contrast, the most equitable imputation would give a profit of $1/n$ to each edge, corresponding to the fractional optimal dual solution which assigns a distance label of $1/n$ to each edge. 

A more extreme example is the following: The graph has $k$ disjoint paths of length $n-1$ from $s$ to $t'$ of unit capacity edges and one more edge, $(t', t)$, of capacity $k$. Since edge $(t', t)$ is a minimum cut, giving the entire profit of $k$ to it is a core imputation. A more  equitable imputation is obtained by setting the distance labels of all edges to $1/n$, giving a profit of $k/n$ to $(t', t)$ and $1/n$ to the rest of the edges. 
\end{example}

Let $F_{\max}$ denote the maximum flow in $G$, i.e., the optimal objective function value of LP (\ref{eq.flow-primal}) and let $E' \subseteq E$ be the set of essential edges per Definition \ref{def.flow-essential}. In Theorem \ref{thm.flow-fair} we prove that LPs (\ref{eq.min-flow}) and (\ref{eq.max-flow}) find min-max fair, max-min fair and equitable  core imputations, respectively, for the max-flow game.

 	\begin{mini}
		{} {{\alpha}}
			{\label{eq.min-flow}}
		{}
		\addConstraint{\sum_{(i, j) \in E}  {c_{ij} \delta_{ij}}}{= F_{\max}}{}
		\addConstraint{\delta_{ij} - \pi_i + \pi_j}{ \geq 0 \quad }{\forall (i, j) \in E}
		\addConstraint{\pi_s - \pi_t}{\geq 1}{}
		\addConstraint{\delta_{ij} \cdot c_{ij}}{\leq \alpha}{\forall (i, j) \in E'}
		\addConstraint{\delta_{ij}}{\geq 0}{\forall (i, j) \in E'}
		\addConstraint{\delta_{ij}}{= 0}{\forall (i, j) \in (E - E')}
		\addConstraint{\pi_i}{\geq 0}{\forall i \in V}
	\end{mini}

\begin{theorem}
	\label{thm.flow-fair}
For the max-flow game:
\begin{enumerate}
\item  An optimal solution to LP (\ref{eq.min-flow}) is a min-max fair core imputation in $C(I)$.
\item  An optimal solution to LP (\ref{eq.max-flow}) is a  max-min fair core imputation in $C(I)$. 
\item  An optimal solution to LP (\ref{eq.equitable-flow}) is an equitable core imputation in $C(I)$. 
\end{enumerate}
\end{theorem}

\begin{proof}
{\bf 1).} Consider an optimal solution $(\delta, \pi)$ to LP (\ref{eq.min-flow}). We will first show that it is an optimal solution to the dual LP (\ref{eq.flow-dual}) as well, i.e. it is an optimal fractional $s$--$t$ cut. By Lemma \ref{lem.flow-essential}, for edge $(i, j) \in (E - E')$, $d_{ij}$ must be zero. This justifies the sixth constraint. Furthermore, by constraints 2, 3, 5 and 7, the solution $(\delta, \pi)$ is a fractional $s$--$t$ cut. Finally, by the first constraint, the solution is an optimal fractional $s$--$t$ cut. 

To complete the proof, observe that the fourth constraint and the minimization of $\alpha$ in the objective ensures that the maximum profit of an essential edge in the optimal solution is  minimized. 

{\bf 2).} Notice that the fourth constraint in LP (\ref{eq.max-flow}) makes the constraint 
$$\delta_{ij} \geq 0 \quad \quad \forall (i, j) \in E'$$ 
redundant and so it has been dropped. Therefore, the only difference between the two LPs is that the first minimizes $\alpha$ and its fourth constraint ensures that the maximum profit of an essential edge is minimized; and the second maximizes $\beta$ and its fourth constraint ensures that the minimum profit of an essential edge is maximized. This completes the proof. 

{\bf 3).} LP (\ref{eq.equitable-flow}) combines the constraints of the previous two LPs and its objective function minimizes $(\alpha - \beta)$, i.e., the spread of the imputation. Therefore, its optimal solution is an equitable core imputation in $C(I)$. 
\end{proof}

 	\begin{maxi}
		{} {{\beta}}
			{\label{eq.max-flow}}
		{}
		\addConstraint{\sum_{(i, j) \in E}  {c_{ij} \delta_{ij}}}{= F_{\max}}{}
		\addConstraint{\delta_{ij} - \pi_i + \pi_j}{ \geq 0 \quad }{\forall (i, j) \in E}
		\addConstraint{\pi_s - \pi_t}{\geq 1}{}
		\addConstraint{\delta_{ij} \cdot c_{ij}}{\geq \beta}{\forall (i, j) \in E'}
		\addConstraint{\delta_{ij}}{= 0}{\forall (i, j) \in (E - E')}
		\addConstraint{\pi_i}{\geq 0}{\forall i \in V}
	\end{maxi}

 	\begin{mini}
		{} {(\alpha -\beta)}
			{\label{eq.equitable-flow}}
		{}
		\addConstraint{\sum_{(i, j) \in E}  {c_{ij} \delta_{ij}}}{= F_{\max}}{}
		\addConstraint{\delta_{ij} - \pi_i + \pi_j}{ \geq 0 \quad }{\forall (i, j) \in E}
		\addConstraint{\pi_s - \pi_t}{\geq 1}{}
		\addConstraint{\delta_{ij} \cdot c_{ij}}{\leq \alpha}{\forall (i, j) \in E'}
		\addConstraint{\delta_{ij} \cdot c_{ij}}{\geq \beta}{\forall (i, j) \in E'}
		\addConstraint{\delta_{ij}}{\geq 0}{\forall (i, j) \in E'}
		\addConstraint{\delta_{ij}}{= 0}{\forall (i, j) \in (E - E')}
		\addConstraint{\pi_i}{\geq 0}{\forall i \in V}
	\end{mini}

\begin{example}
	\label{ex.flow-fair}
For the max-flow game shown in Figure \ref{fig.flow-fair}, the capacities of edges are marked. The maximum flow sends 10 units each on the paths $s, u, t$ and $s, u, v, t$. The distance labels for the min-max fair and max-min fair core imputations for edges $((s, u), (u, t), (u, v), (v, t))$ are $(1/3, 2/3, 1/3,, 1/3)$ and $(1/5, 4/5, 2/5, 2/5)$, respectively, and the profits are $(20/3, 20/3, 10/3,$ $10/3)$ and $(4, 8, 4, 4)$, respectively. The min-max fair imputation is an equitable imputation as well, with a spread of $10/3$. 
\end{example}


\begin{figure}[h]
\begin{center}
\includegraphics[width=4in]{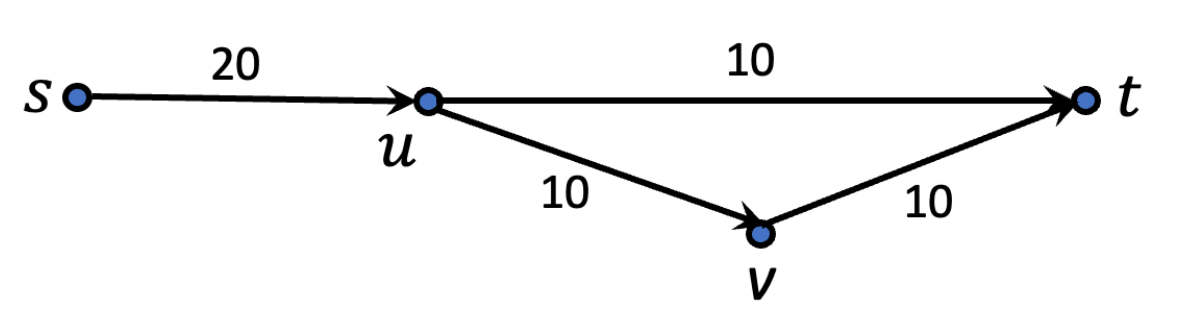}
\caption{The graph for Example \ref{ex.flow-fair}}.  
\label{fig.flow-fair}
\end{center}
\end{figure}


For the max-flow game, one may find it more natural to define min-max fair, max-min fair and equitable  core imputations not with respect to the profit, but with respect to the profit as a {\em percentage} of an agent's capacity. 

For min-max fairness, this can be achieved via the following modification to Definition \ref{def.min-max}: 
$$ p \in \arg \min_{q \in C(I)} \left\{ \max_{e \in E} \left\{ {{q(e)} \over c_e} \right\} \right\} ,$$
and by replacing the fourth constraint of LP (\ref{eq.min-flow}) by: 
$$\delta_{ij}  \leq \alpha \quad \quad \forall (i, j) \in E' .$$

For max-min fairness, the following modification to Definition \ref{def.max-min} is required: 
$$ p \in \arg \max_{q \in C(I)} \left\{ \min_{e \in E'}  \left\{ {{q(e)} \over c_e} \right\} \right\} ,$$
together with replacing the fourth constraint of LP (\ref{eq.max-flow}) by: 
$$\delta_{ij}  \geq \beta \quad \quad \forall (i, j) \in E' .$$

For obtaining an equitable core imputation under this criterion, the definition of {\em spread} of imputation $p$ needs to be changed to the difference between the largest and smallest profit as a {\em percentage} of an essential player's capacity, i.e., 
$$  \spread(p)  :=  (p_{\max} - p_{\min}), $$
$$ \mbox{where} \ \ p_{max} = \left\{ \max_{e \in E'}  \left\{ {{p(e)} \over c_e} \right\} \right\} \ \ \mbox{and} \ \ p_{min} = \left\{ \min_{e \in E'}  \left\{ {{p(e)} \over c_e} \right\} \right\} .$$
The rest of Definition \ref{def.equitable} remains unchanged. The fourth and fifth constraints of LP (\ref{eq.equitable-flow}) need to be replaced by:
$$\delta_{ij}  \leq \alpha \quad \quad \forall (i, j) \in E' $$
and
$$\delta_{ij}  \geq \beta \quad \quad \forall (i, j) \in E' ,$$
respectively.

\section{Discussion}
\label{sec.discussion}

Theorem \ref{thm.Hoffman} applies to an LP whose constraint matrix has entries from the set $\{0, 1, -1\}$. Among the games studied in this paper, only the max-flow game has entries from the set $\{0, 1, -1\}$; in the assignment game and its generalizations, the entries are from the set $\{0, 1\}$.  On the other hand, we note that there are several natural problems in combinatorial optimization which have important applications and for which the entries are from set $\{0, 1, -1\}$, e.g., see \cite{Sch-book}. We leave the problem of studying their core imputations.

An important open question is to shed light on the origins of core imputations, for the bipartite $b$-matching game and the max-flow game, which do not correspond to optimal dual solutions. Is there a ``mathematical structure'' that produces them? Observe that since testing for membership in the core is co-NP-hard for both these games, see Section \ref{sec.related}, there may not be a mathematically clean answer to this question.  

Shapley and Shubik were able to characterize ``antipodal'' points in the core for the assignment game, see Theorem \ref{thm.extreme}. An analogous understanding of the core of the concurrent game will be desirable. For the assignment game, Demange, Gale and Sotomayor \cite{Demange1986multi} give an auction-based procedure to obtain a core imputation; it turns out to be optimal for the side that proposes, as was the case for the deferred acceptance algorithm of Gale and Shapley \cite{GaleS} for stable matching. Is there an analogous procedure for obtaining an imputation in the core of a concurrent game or a $b$-matching game?

\section{Acknowledgements}
\label{sec.ack}

I wish to thank Herv\'e Moulin for asking the interesting question of extending results obtained for the assignment game to general graph matching games having a non-empty core, and Federico Echenique for generously sharing his understanding of issues in economics. I also wish to thank Rohith Gangam and Thorben Trobst for several valuable discussions.

	\bibliographystyle{alpha}
	\bibliography{refs}

\end{document}